\documentclass[sigplan,nonacm,screen]{acmart}

\usepackage[ruled, vlined, linesnumbered]{algorithm2e}
\usepackage{graphicx}
\usepackage{textcomp}
\usepackage{xcolor}
\usepackage{tikz}
\usetikzlibrary{quantikz2,arrows.meta,backgrounds,calc,fit,positioning}
\usepackage{hyperref}
\usepackage{makecell}
\usepackage{diagbox}
\usepackage{adjustbox}
\usepackage{cleveref}
\crefname{theorem}{theorem}{theorems}
\Crefname{theorem}{Theorem}{Theorems}
\crefname{lemma}{lemma}{lemmas}
\Crefname{lemma}{Lemma}{Lemmas}
\crefname{proposition}{proposition}{propositions}
\Crefname{proposition}{Proposition}{Propositions}
\crefname{corollary}{corollary}{corollaries}
\Crefname{corollary}{Corollary}{Corollaries}
\AddToHook{env/theorem/before}{\crefalias{theorem}{theorem}}
\AddToHook{env/lemma/before}{\crefalias{theorem}{lemma}}
\AddToHook{env/proposition/before}{\crefalias{theorem}{proposition}}
\AddToHook{env/corollary/before}{\crefalias{theorem}{corollary}}
\usepackage{multirow}
\usepackage{enumitem}
\usepackage{soul}
\usepackage[most]{tcolorbox}
\usepackage{fancyhdr}
\usepackage{array}
\usepackage{pifont}
\usepackage{mwe}
\usepackage{orcidlink}
\usepackage{fontawesome5}
\usepackage[caption=false,font=small,labelfont=rm,textfont=rm]{subfig}

\definecolor{SkyBlue}{RGB}{100, 180, 220}
\definecolor{TextHighlightBlue}{RGB}{72,120,208}

\newcommand{\note}[1]{{\color{TextHighlightBlue} #1}}
\newcommand{\appendixnote}[1]{{#1}}
\newcommand{\ZY}[1]{{\color{purple}[ZY: #1]}}

\newcommand{\dquote}[1]{``#1''}
\newcommand{\code}{\texttt}

\newcommand{\phoenix}{\textsc{Phoenix}}

\newcommand{\symphony}{\textsc{Symphony}}

\newcommand{\qiskit}{\textsc{Qiskit}}
\newcommand{\tket}{\textsc{Tket}}
\newcommand{\paulihedral}{\textsc{Paulihedral}}

\newcommand{\tetris}{\textsc{Tetris}}
\newcommand{\quclear}{\textsc{Quclear}}
\newcommand{\rustiq}{\textsc{Rustiq}}
\newcommand{\pcoast}{\textsc{Pcoast}}

\newcommand{\totalWeight}{w_\mathrm{tot.}}

\definecolor{ElegantGrayBack}{RGB}{242, 242, 242} % 优雅的浅灰背景
\definecolor{ElegantGrayFrame}{RGB}{105, 105, 105} % 配套的深灰边框

\newtcolorbox{takeaways}[1][]{
    colback=ElegantGrayBack, % 背景色
    colframe=ElegantGrayFrame, % 边框色
    fonttitle={\bfseries},
    title={Takeaways},       % 默认标题
    arc=1.5mm,                 % 圆角半径
    boxrule=0.8pt,           % 边框粗细
    left=5pt, right=5pt,     % 左右边距
    top=3pt, bottom=3pt,     % 上下边距
    enhanced,                % 启用增强样式
    breakable,               % 允许分页
    #1                       % 可选参数（可覆盖默认样式）
}

\SetAlFnt{\small}

\begin{document}

% \title{
% Global Synthesis for Hamiltonian Simulation Using Binary Symplectic Form
% \thanks{{\faUserFriends} These authors contributed equally to this work.}
% }

\title{
% Symphony: 
Efficient Compilation for Hamiltonian Simulation via Global Binary Symplectic Form Simplification
% \thanks{{\faUserFriends} These authors contributed equally to this work.}
}

\author{Zhaohui Yang}
\authornote{Both authors contributed equally to this work.}
\orcid{0000-0003-4698-4378}
% \author{G.K.M. Tobin}
% \authornotemark[1]
% \email{webmaster@marysville-ohio.com}
\affiliation{%
  \institution{The Hong Kong University of\\ Science and Technology}
  \city{Hong Kong}
  % \state{Hong Kong}
  \country{}
}
\email{zhaohui@ucsb.edu}

\author{Yuwei Han}
\authornotemark[1]
\orcid{0009-0000-2522-5909}
\affiliation{
  \institution{Tsinghua University}
  \city{Beijing}
  \country{China}
}
\email{hanyw22@mails.tsinghua.edu.cn}

\author{Ruiyun Zhang}
\orcid{0009-0001-4878-4372}
\affiliation{
  \institution{The Hong Kong University of\\ Science and Technology}
  \city{Hong Kong}
  \country{}
}
\email{eeryzhang@ust.hk}

\author{Dawei Ding}
\orcid{0000-0001-7728-5380}
\affiliation{
  \institution{Fudan University}
  % \institution{Shanghai Institute for Mathematics and Interdisciplinary Sciences}
  \city{Shanghai}
  \country{China}
  \country{}
}
\email{daweiding@fudan.edu.cn}

\author{Jianxin Chen}
\orcid{0000-0002-9365-776X}
\affiliation{
  \institution{Tsinghua University}
  \city{Beijing}
  \country{China}
}
\email{chenjianxin@tsinghua.edu.cn}

\author{Yuan Feng}
\orcid{0000-0002-3097-3896}
\affiliation{
  \institution{Tsinghua University}
  \city{Beijing}
  \country{China}
}
\email{yuan_feng@tsinghua.edu.cn}

\author{Yuan Xie}
\orcid{0000-0003-2093-1788}
\affiliation{
  \institution{The Hong Kong University of\\ Science and Technology}
  \city{Hong Kong}
  \country{}
}
\email{yuanxie@ust.hk}

\renewcommand{\shortauthors}{Zhaohui Yang et al.}

%%%%%%%%%%%%%%%%%%%%%%%%%%%%%%%%%%%%%%%%
%%%%%%%% -- PAPER CONTENT STARTS -- %%%%%%%%%

\begin{abstract}
    
Hamiltonian simulation is a core quantum workload, underpinning variational quantum algorithms and Trotterized time evolution. Such programs are expressed as Pauli exponential sequences, exhibiting structural patterns that are highly amenable to high-level synthesis and optimization.
% Such programs are expressed as variably arranged sequences of Pauli exponentials, with structural patterns suitable for high-level synthesis and optimization. 
Existing compilers, however, fail to fully unlock the optimization potential of their global algebraic structure, even when employing advanced graph- or tableau-based methods. 

We present \symphony, a holistic compilation approach built on the binary symplectic form (BSF) representation of Pauli strings.
% \symphony\ aggressively optimizes both two-qubit gate count and circuit depth, primarily targeting near-term and early fault-tolerant quantum devices. 
Unlike prior group-wise BSF simplification and path-based Pauli network synthesis, \symphony\ applies generalized controlled-Pauli Clifford transformations directly to a global BSF tableau, progressively reducing active Pauli rows while adaptively emitting eligible two-qubit blocks in addition to single-qubit rotations in a forward Clifford frame. Following algebraic simplification, \symphony\ performs a causality-preserving block rescheduling heuristic that respects frame-induced dependencies while exposing extensive two-qubit block parallelism opportunities. This streamlined compilation style comprehensively exploits simultaneous simplification and commutativity opportunities, achieving efficient global optimization without relying on computationally expensive heuristics or long-horizon searches. Across the generic Hamiltonian simulation benchmarks in HamLib, \symphony\ achieves average reductions of 59\% in two-qubit gate count and 91\% in circuit depth. It Pareto-dominates prior state-of-the-art compilers, requiring 1.17--1.58$\times$ fewer two-qubit gates and especially shrinking two-qubit circuit depth by a substantial factor of 1.87--5.67$\times$ on average. 
\end{abstract}

% \keywords{Quantum Computing, Hamiltonian Simulation, Quantum Compilation, Compiler, Global Optimization}

\maketitle

\section{Introduction}\label{sec:introduction}

Quantum computing promises transformative speedups for classically intractable problems, ranging from prime factorization~\cite{shor1994algorithms} and linear systems~\cite{harrow2009quantum} to the simulation of complex physical systems~\cite{feynman1982simulating,lloyd1996universal}. Hamiltonian simulation has emerged as one of the most compelling applications of quantum computers~\cite{feynman1982simulating,lloyd1996universal,georgescu2014quantum}, underpinning workloads in quantum chemistry~\cite{mcardle2020quantum,cao2019quantum}, condensed-matter and many-body physics~\cite{altman2021quantum}, materials science~\cite{bauer2020quantum}, and combinatorial optimization~\cite{abbas2024challenges}. Crucially, its centrality persists across algorithmic paradigms: NISQ formulations such as Trotterized time evolution~\cite{childs2018toward} and variational frameworks including VQE~\cite{peruzzo2014variational} and QAOA~\cite{farhi2014quantum} share the same computational core with fault-tolerant approaches based on product formulas, qubitization, and quantum signal processing~\cite{su2021fault,low2019hamiltonian}.

Regardless of the algorithmic paradigm, compiling a Hamiltonian simulation workload reduces to a common structural problem: constructing efficient quantum circuits to approximate a desired unitary evolution under a system Hamiltonian. Hamiltonians are represented as weighted sums of Pauli strings while unitary evolutions are required to be lowered into or approximated by sequences of parameterized Pauli rotations $\left\{e^{-i \theta_k P_k}\right\}$~\cite{dalzell2023quantum}. Despite originating from diverse applications and algorithmic formulations, these rotations expose a common compilation interface that we refer to as the Pauli intermediate representation (\emph{Pauli-IR}). % throughout this paper.
% Efficiently synthesizing such Pauli-IR sequences is therefore a recurring compilation problem across Hamiltonian-simulation workloads.

\begin{figure}[tbp]
    \centering
    \subfloat[Na\"ive synthesis in the CNOT-tree style]{
        \includegraphics[width=\columnwidth]{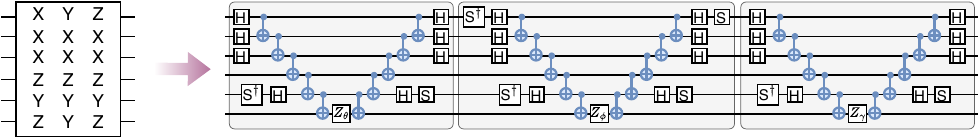}
        \label{fig:naive-synthesis}
    }
    \hfil
    \subfloat[Simultaneous simplification through Clifford conjugation]{
        \includegraphics[width=\columnwidth]{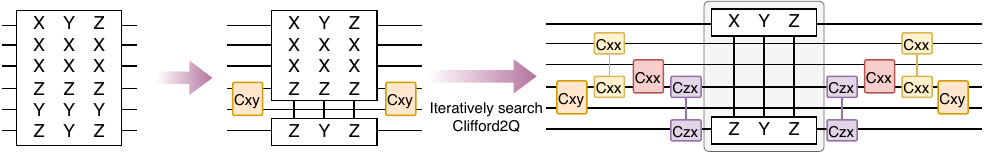}
        \label{fig:simultaneous-simplification}
    }
    \caption{(a) Na\"ive synthesis vs. (b) simultaneous simplification of Pauli exponentials.}
    % \caption{Pauli exponential synthesis approach comparison.}
    \label{fig:naive-vs-simultaneous-synthesis}
\end{figure}

Despite this compact high-level representation, straightforward synthesis of Pauli-IRs produces circuits with prohibitive gate counts and depths. Conventionally, each Pauli rotation or Pauli exponential $e^{-i\theta P}$ is realized as $C\, R_Z(2\theta)\, C^\dagger$, where the Clifford $C$ comprises single-qubit basis changes followed by a CNOT ladder or CNOT parity tree, such that a weight-$w$ Pauli exponential requires $2(w-1)$ CNOT gates for synthesis~\cite{nielsen2010quantum}. However, a na\"ive, term-by-term synthesis inherently ignores the global algebraic commonalities of the Hamiltonian. It forces the circuit to repeatedly compute and uncompute overlapping parity structures rather than leveraging the shared relationships across the Pauli sequence. Such a synthesis approach is exemplified in \Cref{fig:naive-synthesis}.

Prior work has explored various high-level synthesis and optimization strategies, ranging from localized inter-block gate cancellation~\cite{li2022paulihedral,jin2024tetris,li2025pauliforest} and structural abstraction rewriting~\cite{van2020circuit,mukhopadhyay2023synthesizing,cowtan2019phase,paykin2023pcoast,liu2025quclear} to global algebraic simplifications over binary symplectic form (BSF) tableaux~\cite{goubault2024faster,yang2025phoenix}.  An illustrative example of simultaneous simplification of Pauli exponentials via Clifford conjugation, as employed in related compilation frameworks, is shown in \Cref{fig:simultaneous-simplification}. However, the most closely related tableau- and Pauli-network-based approaches remain constrained by strict adherence to a \dquote{first-diagonalize, then-search-CNOT} methodology~\cite{goubault2024faster,kuo2026unified}, non-holistic grouping strategies and computationally expensive heuristic metrics~\cite{yang2025phoenix}, or long-horizon searches~\cite{machiya2026monteq,dubal2025paulinetwork}. Collectively, these limitations point to the need for a streamlined synthesis approach that combines broad algebraic scope with computationally efficient optimization.

% In this work, we propose \symphony, an efficient compilation framework for generic Hamiltonian simulation programs built on the holistic algebraic simplification of the BSF tableau representation of Pauli-IRs. While adopting the same Clifford-based search space used in \citet{yang2025phoenix}, we instead process all Pauli-IRs within a single tableau. Furthermore, we perform aggressive two-qubit block emission rather than single-qubit rotation peeling, which carries a higher risk of circuit degradation. By delving into the algebraic structures of Pauli-IRs and Clifford transformations, we exhaustively exploit opportunities for simultaneous simplification and commutativity optimization. Even when guided by a lightweight, localized selection of appropriate Cliffords, \symphony\ proves to be highly-effective global optimization approach without the need for long-horizon searches. Our key contributions are summarized as follows:
In this work, we propose \symphony, an efficient compilation framework for generic Hamiltonian simulation programs. To overcome the limitations of existing tableau-based approaches, \symphony\ abandons fragmented grouping and restrictive diagonalization paradigms in favor of a holistic algebraic simplification over a unified BSF tableau. Even when guided by lightweight, localized Clifford selections, \symphony\ effectively unlocks simultaneous simplification without resorting to complex heuristics or long-horizon searches. Our key contributions are summarized as follows:
% \begin{itemize}[leftmargin=*, topsep=2pt, itemsep=2pt, parsep=2pt]
\begin{itemize}[leftmargin=*]
    % \item \emph{Holistic tableau simplification formulation.} 

    \item We formulate Pauli-IR compilation as holistic BSF tableau simplification. By evaluating CNOT-equivalent controlled-Pauli Cliffords against all unresolved Pauli strings, we unlock global optimization opportunities often obscured by prior grouping or diagonalization methods.

    % \item \emph{Efficient forward-frame synthesis.}
    \item We develop a guaranteed-progress, forward-frame algebraic simplification algorithm that eliminates tableau weights row-by-row. Our greedy Clifford selection maximizes whole-tableau weight reduction, proving more effective than complex heuristics or legacy path-based synthesis.

    % while adaptively emitting Pauli-IR blocks of weight at most two
    \item During simplification, we adaptively emit two-qubit blocks rather than just single-qubit rotations, mitigating circuit degradation risks while exposing additional parallelism for downstream scheduling.

    % This strategy inherently captures simultaneous simplification across heterogeneous-weight Pauli strings, proving significantly more effective and intuitive than complex heuristic-based methods or legacy path-based Pauli network synthesis.
    % \item \emph{Commutativity-optimized ASAP scheduling.} 
    \item We introduce a causality-preserving ASAP block rescheduling heuristic inspired by graph edge-coloring. By exploiting inherent sequence flexibility and canonical commutation relations, it preserves frame-induced precedences while maximizing two-qubit block parallelism.
    
    % exposing rich two-qubit parallelism to aggressively optimize circuit depth.

    % We minimize circuit depth by mapping the peeled Pauli-IRs and synthesized Cliffords to a dependency graph for as-soon-as-possible (ASAP) scheduling. By leveraging the controlled-Pauli Clifford formalism, we expose canonical commutation relations among Pauli operators and Cliffords, enabling relaxed commutativity constraints and deeper optimization during scheduling.
\end{itemize}

% \begin{itemize}[leftmargin=*]
%     \item We formulate Pauli-IR compilation as the holistic simplification of a global BSF tableau. By using CNOT-equivalent controlled-Pauli Clifford gates as simplification primitives, we evaluate each Clifford across all unresolved Pauli strings, revealing global optimization opportunities typically hidden by grouping- or diagonalization-based approaches.
%     \item We develop a guaranteed-progress, forward-frame algebraic simplification algorithm that systematically eliminates tableau weights row-by-row. Our lightweight, greedy Clifford selection evaluates operations based on their whole-tableau weight reduction and proves more effective and intuitive than complex heuristic-based or legacy path-based Pauli network synthesis.
%     \item We utilize an adaptive two-qubit block emission mechanism during simplification, to mitigate the risk of degradation associated with emitting only single-qubit rotations and expose additional parallelism to downstream scheduling.
%     \item We introduce a causality-preserving ASAP block rescheduling, implemented as an graph edge-coloring-inspired heuristic. It preserves frame-induced precedences while exploiting canonical commutation relations among program primitives and permitted Trotter-block ordering freedom to expose rich two-qubit block parallelism and thus aggressively optimizes circuit depth.
% \end{itemize}

We evaluate \symphony\ on 100 representative Hamiltonians from HamLib---recently utilized in Qiskit performance benchmarking~\cite{nation2025benchmarking}---alongside selected scaled UCCSD workloads, covering diverse application domains, problem sizes, hardware connectivities, and both near-term and early fault-tolerant cost models. Specifically, relative to na\"ive term-wise synthesis across HamLib, \symphony\ reduces two-qubit gate counts and circuit depth by an average of 59\% and 91\%, respectively. Against prior state-of-the-art (SOTA) compilers,
% \symphony\ achieves a 1.05- to 1.44-fold suppression in two-qubit gate count and a 1.79- to 5.13-fold improvement in two-qubit circuit depth on average.
\symphony\ yields improvements of 1.17--1.58$\times$ in two-qubit gate count and 1.87--5.67$\times$ in circuit depth on average.

\section{Background}\label{sec:background}

\subsection{Hamiltonian Simulation Programs}\label{sec:background-hamiltonian-simulation}

Hamiltonian simulation programs form an important workload class in chemistry, materials, and optimization~\cite{dalzell2023quantum}. After a system Hamiltonian is mapped onto $n$ qubits, it can generally be written as a real linear combination of Pauli operators, $H=\sum_{k=1}^{m} c_k P_k$, $P_k\in\{I,X,Y,Z\}^{\otimes n}$. The corresponding simulation task is to approximate the unitary evolution $U(t)=e^{-iHt}$ using quantum circuits comprising basic single-qubit and two-qubit gates.

Since Pauli terms in $H$ generally do not commute, product-formula methods approximate $U(t)$ by a sequence of Pauli exponentials~\cite{trotter1959product,suzuki1990fractal}. 
Pauli-exponential sequences also naturally arise in variational quantum algorithms, such as QAOA and Trotterized unitary coupled-cluster (UCC) ans\"atze, which lower excitation generators to parameterized Pauli rotations~\cite{farhi2014quantum,romero2018strategies}. Thus, regardless of whether the source is real-time evolution or a variational ansatz, these workloads share the same high-level representation: a sequence of Pauli exponentials. Their admissible ordering transformations simply depend on the source semantics (e.g., mutually commuting rotations can typically be reordered exactly).

We therefore represent the compiler input as a sequence of Pauli rotations denoted as Pauli-IRs,
\begin{align}
    \mathcal{P}
    =
    \bigl\{
        (P_k,\theta_k)
    \bigr\}_{k=1}^{m},
    \quad
    R_{P_k}(\theta_k)
    =
    e^{-i\theta_k P_k},
\end{align}
together with the ordering freedom exposed by the front end. For the product-formula workloads considered in this work, the compiler may select an admissible term order within each
simulation step; exact commutation constraints are retained when required by the input semantics. 

\subsection{Binary Symplectic Form of Pauli-IRs}\label{sec:background-bsf}

\symphony\ represents Pauli operators in a high-level, formal binary symplectic form (BSF) tableau~\cite{aaronson2004improved,maslov2018shorter}. 
% The BSF is a standard representation for Pauli and Clifford computation~\cite{aaronson2004improved,maslov2018shorter}. 
Given $m$ Pauli rotations on $n$ qubits, their operators are encoded as
\begin{align}
    T=[X\,|\,Z]\in\mathbb{F}_2^{m\times 2n},
\end{align}
where the $i$-th row encodes Pauli operator $P_i$ as a binary vector $p_i=(x_i\,|\,z_i)$. The bit pairs
\begin{align}
    (X_{i,q},\,Z_{i,q}) = (0,0),\,(1,0),\,(0,1),\,(1,1)
\end{align}
indicate $I$, $X$, $Z$, and $Y$, respectively. Rotation angles are stored separately from the tableau. 
The support and weight of a Pauli row follow directly from
the binary representation:
\begin{align}
    s_i= X_{i,:} \lor Z_{i,:},
    \quad
    w_i=\operatorname{wt}(P_i)
       =\lVert X_{i,:}\lor Z_{i,:}\rVert_1.
\end{align}
Under conventional parity-tree synthesis, a nontrivial weight-$w_i$ Pauli exponential individually requires $2(w_i-1)$ CNOT gates. Reducing BSF weights is therefore a direct high-level proxy for reducing synthesis overhead.

BSF also exposes Pauli commutation algebraically. Two rows $p_i=(x_i\,|\,z_i)$ and $p_j=(x_j\,|\,z_j)$ commute if and only if their binary symplectic inner product vanishes:
\begin{align}
    [P_i,P_j]=0
    \quad\Longleftrightarrow\quad
    x_i z_j^{T}+z_i x_j^{T}=0
    \pmod 2.
    \label{eq:pauli-commutation}
\end{align}
Thus, support, weight, overlap, and commutation can all be computed through Boolean and arithmetic operations over $\mathbb{F}_2$, without constructing the exponentially large matrix representation of any Pauli operator.

% In this paper, each Pauli operator is stored as a
% \emph{row}; consequently, we consistently refer to
% $w_i$ as a BSF row weight. A transposed BSF convention would
% instead describe the same quantity as a column weight.

% \note{
% \symphony\ represents the Pauli operators in this high-level IR using the binary symplectic form (BSF), a standard tableau representation for Pauli and Clifford computation~\cite{aaronson2004improved,maslov2018shorter}.
% Given $m$ Pauli rotations on $n$ qubits, their Pauli operators are stored as
% \begin{equation}
%     T=[X|Z]\in\{0,1\}^{m\times 2n},
% \end{equation}
% where row $i$ encodes the Pauli operator $P_i$. For qubit $q$, the bit pair $(X_{i,q},Z_{i,q})=(0,0),(1,0),(0,1),(1,1)$ denotes $I,X,Z,Y$, respectively. The rotation angle is stored separately from the BSF row.

% The BSF exposes row weights through simple Boolean operations. For row $i$, the weight
% \begin{equation}
%     w_i=\left\|X_{i,:}\lor Z_{i,:}\right\|_1
%     =\sum_{q=1}^{n}(X_{i,q}\lor Z_{i,q})
% \end{equation}
% counts the number of qubits touched by the Pauli row. Rows with smaller weights require fewer entangling operations to lower, so weight reduction becomes the main tableau-level objective.

% The following subsection introduces the Clifford transformations used to update the BSF tableau.
% }

\subsection{Clifford Formalism}\label{sec:background-clifford}

\begin{figure}[tbp]
    \centering
    \includegraphics[width=\columnwidth]{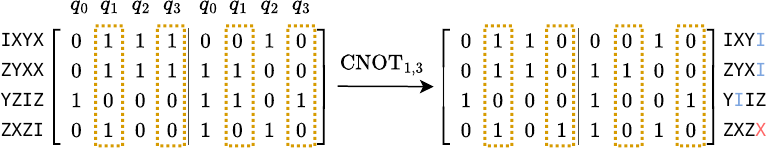}
    \caption{Example of BSF update by Clifford transformation.}
    \label{fig:cnot-update-example}
\end{figure}

The Clifford group is defined as the normalizer of the Pauli group. Thus, conjugating a Pauli operator $P$ by a Clifford operator $C$ yields another Pauli operator $P'$ up to a sign:
\begin{align}
    C P C^\dagger = sP', \quad s\in\{-1,+1\}.
\end{align}
It inherently preserves the Pauli-rotation form as
\begin{align}
    C e^{-i\theta P} C^\dagger =e^{-i(s\theta)P'}.
\end{align}
The induced sign $s$ can either be tracked separately or absorbed into the Pauli exponential's rotation angle.

In the BSF representation, Clifford conjugation acts as an invertible symplectic transformation over $\mathbb{F}_2$. A single-qubit Clifford permutes the three nonidentity Pauli letters on one qubit and therefore cannot change Pauli weight. In contrast, a two-qubit Clifford may alter whether one or both acted-on qubits belong to a Pauli support. When applied to qubits $(a,b)$, it updates only the BSF columns $(X_a,\,X_b,\, Z_a,\, Z_b)$, but the update is performed simultaneously for every Pauli row. For example, the tableau update rule by $\mathrm{CNOT}$ follows
\begin{align}
    [X_a,\,X_b\,|\,Z_a,\,Z_b] \rightarrow [X_a,\,X_a\oplus X_b\,|\,Z_a\oplus Z_b, \,Z_b]
\end{align}
and \Cref{fig:cnot-update-example} demonstrates its action on a whole tableau.

We specifically leverage a controlled-Pauli Clifford family formally known as the \emph{universal controlled gate} (UCG)~\cite{grier2022classification}:
\begin{equation}
    C(P,Q) = \frac{1}{2}((I+P)\otimes I + (I-P)\otimes Q),
\end{equation}
where $P,Q\in\{X,Y,Z\}$ specify its two Pauli axes. Any $C(P,Q)$ is Hermitian and self-inverse; the nine choices of $C(P,Q)$ together constitute a set of generators for the two-qubit Clifford group~\cite{grier2022classification}. 
Because these operations are CNOT/CZ-equivalent up to single-qubit Cliffords as
\begin{center}
\begin{adjustbox}{width=0.48\columnwidth}
    \begin{quantikz}[row sep=0.2cm, column sep=0.2cm, align equals at=1.5]
        & \gate[2,disable auto height,style]{C(P,Q)} & \ghost{H_Q}\qw \\
        & & \ghost{H_Q}\qw\\
    \end{quantikz}\,=\begin{quantikz}[row sep=0.2cm, column sep=0.2cm, align equals at=1.5]
        & \gate{H_{P}} & \ctrl{1} & \gate{H_P} & \ghost{H_Q}\qw \\
        & \gate{H_{Q}} & \ctrl{-1} & \gate{H_Q} & \ghost{H_Q}\qw
    \end{quantikz},
\end{adjustbox}
\end{center}
where $H_Z=I$, $H_X=H$, and $H_Y=S H S^\dagger$, each choice $C(P,Q)$ follows a uniform cost profile during high-level synthesis. Their distinct Pauli axes therefore induce fundamentally different BSF update trajectories, exposing uniquely varied opportunities for simultaneous weight reduction. \appendixnote{A more comprehensive discussion of the UCG Clifford formalism is provided in Appendix~\ref{sec:appendix_ucg_clifford}.}

\iffalse
\subsection{Existing Solutions}

Existing Pauli-based compilers can be distinguished by their primary optimization objects. Gate-cancellation methods optimize the synthesis variants and execution orders of Pauli rotations to expose cancellation opportunities~\cite{li2022paulihedral,jin2024tetris}. Methods based on graphs or diagrams encode Pauli operations and their dependencies for algebraic rewriting~\cite{cowtan2019phase,paykin2023pcoast}. Methods based on Pauli networks or BSF tableaux directly manipulate algebraic Pauli representations through Clifford transformations \cite{goubault2024faster,schmitz2024graph,yang2025phoenix}. These categories are not mutually exclusive but characterize the dominant abstraction used by each compiler. Detailed comparisons are presented in \Cref{sec:related_work}.
\fi

\section{Motivation}\label{sec:motivation}

% \subsection{Bottlenecks of Existing Solutions}

\subsection{Limitations of Existing Pauli-IR Synthesis}

Prior block-, cluster-, and graph-based compilers expose substantial cross-term cancellation and shared-basis synthesis opportunities, but their optimization scope is commonly bounded by predetermined blocks, commuting clusters, or extracted circuit regions~\cite{li2022paulihedral,jin2024tetris,li2025pauliforest,cowtan2019phase,paykin2023pcoast,liu2025quclear}. More recent approaches formalize Pauli-IR compilation as a global algebraic optimization over the binary symplectic form (BSF) tableau~\cite{goubault2024faster,yang2025phoenix}. Despite this advancement, existing tableau-based methods remain limited: they often adhere to restrictive \dquote{first-diagonalize, then-search-CNOT} paradigms~\cite{goubault2024faster,kuo2026unified}, employ fragmented support-based grouping~\cite{yang2025phoenix}, rely on computationally expensive metrics~\cite{yang2025phoenix}, or demand prohibitively long-horizon searches~\cite{machiya2026monteq,dubal2025paulinetwork}. Consequently, the rich algebraic properties of the Clifford formalism remain underexploited.

\subsection{Design Space of Holistic BSF Simplification}\label{sec:motivation-design-space}

% \subsection{BSF as Natural Representation for Global Optimization}

% Say why BSF serves natural representation encoding global information of Pauli strings}
% We affirm that BSF servers as natural and formal representation of Pauli-IRs, with global information stored}
% particularly for scope of logical-level synthesis and variably arranged Pauli strings}

\paragraph{BSF as Natural Representation for Global Optimization}
\iffalse
To effectively optimize Pauli-IR sequences, it is imperative to choose a representation that captures the underlying algebraic relationships between Pauli operators and Clifford transformations. Conventional synthesis paradigms track quantum operations through circuit-level structures such as CNOT parity trees or phase gadgets~\cite{li2022paulihedral,jin2024tetris,cowtan2019phase}. These structures describe one pre-diagonalized operator at a time, so the correlations that make a Hamiltonian compressible---shared supports, matching Pauli axes, and mutual commutation---are dissolved into gate sequences before they can be exploited.

The BSF tableau introduced in \Cref{sec:background-bsf} preserves exactly this information. All $m$ Pauli operators occupy a single $\mathbb{F}_2^{m\times 2n}$ tableau, and the quantities that govern synthesis cost are read directly from it: row weight determines the $2(w_i-1)$ CNOT cost of conventional parity-tree synthesis.
% ; the symplectic inner product in \eqref{eq:pauli-commutation} decides commutation without exponential-size matrix multiplication.
A two-qubit Clifford conjugation modifies only the four tableau columns of its operand pair, yet it updates every row simultaneously.
Consequently, by modeling the state of all Pauli operators simultaneously within a unified tableau, BSF exposes global structural patterns and enables synthesis engines to reason about transformations across the entire sequence, rather than being confined to local, peephole scopes.
\fi

Term-wise CNOT parity-tree synthesis processes one pre-diagonalized operator at a time and therefore dissolves Hamiltonian-wide correlations---such as shared supports and mutual commutation---before they can be leveraged. Block- and phase-gadget-based methods can instead jointly optimize multiple rotations~\cite{li2022paulihedral,jin2024tetris,cowtan2019phase}; however, when applied through selected blocks or extracted circuit regions, cross-term information beyond those boundaries remains inaccessible. In contrast, the unified BSF tableau introduced in \Cref{sec:background-bsf} naturally preserves these algebraic relationships across the active Pauli sequence. A single two-qubit Clifford conjugation modifies only four tableau columns but simultaneously updates every Pauli row. Thus, BSF natively exposes global structural patterns, empowering synthesis engines to explore broad optimization pathways rather than being confined to local, peephole scopes.

% \subsection{Necessity and Complexity of Holistic Synthesis}

%This subsection is crucial. Illustrated by the following paragraphs.}
%Following specific motivations/observations guide this work's technical design, philosophy, and evaluation}

\paragraph{Optimization space induced by simplification primitives}

Recent path-based Pauli network synthesis methods~\cite{goubault2024faster,kuo2026unified} largely mirror phase polynomial synthesis~\cite{vandaele2022phase,kissinger2019cnot}: they first diagonalize a target row, then eliminate its support using directional CNOTs. However, unlike phase tables where CNOTs execute independent column additions, a BSF tableau dictates symplectic updates where the $X$ and $Z$ components couple. Consequently, the weight reduction of any given operation strictly depends on preceding Clifford transformations. Static graph-matching models misrepresent this dynamically evolving objective and artificially bound the reachable optimization space.

Furthermore, directional CNOTs overly restrict the Clifford search space. We instead adopt the CNOT-equivalent UCG family $C(P,Q)$ introduced in \Cref{sec:background-clifford}~\cite{grier2022classification,yang2025phoenix}, which intrinsically encodes basis choice and entangling direction. For any weight-$2$ interaction (e.g., $ZZ$), exactly four of the nine UCGs will directly reduce it to weight-$1$. By abandoning the rigid pre-diagonalization paradigm, UCGs act natively on the BSF tableau, guaranteeing immediate weight reduction for any non-trivial two-qubit operator without the computational overhead of $\{H,S,\mathrm{CNOT}\}$ long-horizon searches~\cite{dubal2025paulinetwork,machiya2026monteq}.

\paragraph{Optimization opportunities missed by grouping-based simplification}

\begin{figure}[tbp]
    \centering
    \subfloat[Group-wise BSF simplification]{
        \includegraphics[width=\linewidth]{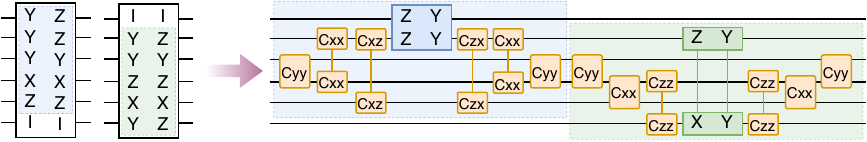}
        \label{fig:grouping-simplification}
    }\hfil
    \subfloat[Holistic BSF simplification]{
        \includegraphics[width=\linewidth]{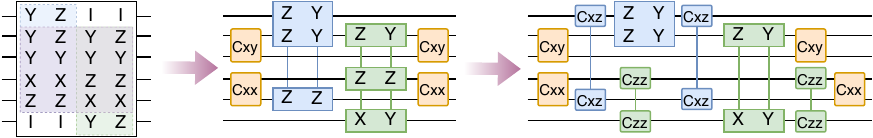}
        \label{fig:holistic-simplification}
    }
    \caption{Comparison of (a) group-wise and (b) holistic BSF simplification. By evaluating Clifford transformations against the entire tableau, the holistic approach captures cross-support overlaps, thereby halving the required UCG conjugations in this example.}
    \label{fig:grouping-vs-holistic-simplification}
\end{figure}

Another limitation commonly occurring in prior tableau-based synthesis methods is that first partitioning Pauli strings into groups of identical support or mutually commuting terms and simplifying each group in isolation~\cite{yang2025phoenix,van2020circuit} leave broader global optimization opportunities. 
% As demonstrated in \Cref{fig:grouping-vs-holistic-simplification}, a group-wise approach is blind to overlapping subset patterns that exist across heterogeneously supported Pauli strings, thereby squandering potential simultaneous simplifications.
\Cref{fig:grouping-vs-holistic-simplification} illustrates this loss concretely: the four heterogeneously supported strings share exploitable structure that spans the imposed group boundaries, and processing them without grouping halves the number of required two-qubit Clifford conjugations. 
In real-world applications, cross-group optimization opportunities are ubiquitous. For example, in chemistry simulations, Jordan-Wigner (JW) encodings characteristically yield regular, chain-like Pauli patterns, whereas Bravyi-Kitaev (BK) encodings produce sparse, heterogeneous tree-like ones. A holistic formulation retains these opportunities. It additionally dispenses with the inter-group ordering pass that grouped pipelines require to re-stitch their blocks~\cite{yang2025phoenix,tomesh2021optimized}.

\subsection{Beyond Single-Qubit-Only Emission}\label{sec:motivation-beyond-1q}

\begin{table}[tbp]
    \centering
    \caption{Average synthesis cost in CNOT count of $e^{-i\sum_j \theta_j P_j}$ for generic coefficients, averaged over all fixed-cardinality subsets $\mathcal{P}\subseteq \{X,Y,Z\}^{\otimes 2}$.}
    \label{tab:pauli-exp-synth-cost}
    \small\begin{tabular}{|c|c|c|c|c|c|c|c|}
\hline
$\lvert \mathcal{P} \rvert$ & 3 & 4 & 5 & 6 & 7 & 8 & 9 \\
\hline
$\#\mathrm{CNOT}$ & 2.07 & 2.29 & 2.64 & 2.93 & 3.0 & 3.0 & 3.0 \\
\hline
\end{tabular}

\end{table}

\begin{figure}[tbp]
    \centering
    \includegraphics[width=\linewidth]{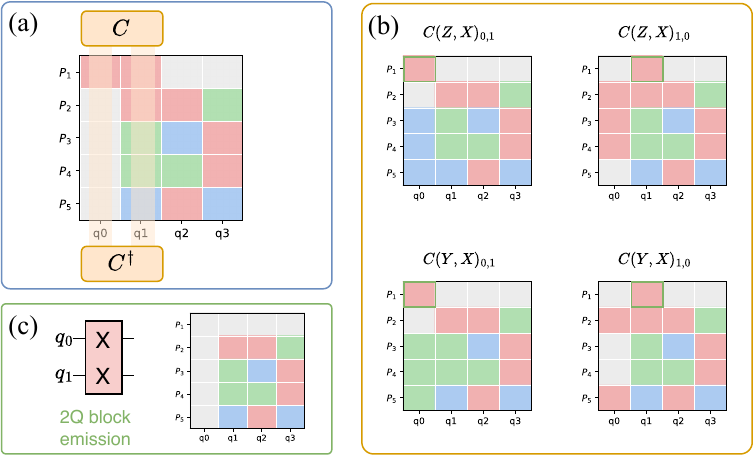}
    \caption{Avoiding degradation via early two-qubit block emission during BSF simplification ($X$, $Y$, $Z$ shown as red, green, blue). (a) Initial Pauli strings to be simplified by Clifford conjugation acting on ($q_0, q_1$). (b) Forcing target $P_1$ to weight-$1$ invariably inflates the weights of remaining strings. (c) Emitting $P_1$ early at weight-$2$ prevents collateral inflation and significantly reduces overall synthesis overhead.}
    % \caption{Beyond single-qubit rotation emission during Pauli terms simplification. $X$, $Y$, and $Z$ operators are represented by red, green, and blue squares, respectively. (a) A set of heterogeneous-weight Pauli strings to be simplified by appropriate Clifford conjugation on the operands ($q_0, q_1$) of the lowest-weight Pauli string $P_1$. (b) Examples of four beneficial Clifford selections reduce $P_1$ to be weight-$1$, but inevitably all increase the weights of some other Pauli strings. (c) An alternative early-emission at weight-$2$ for $P_1$ results in a two-qubit Pauli rotation and remaining Pauli strings, with lower overall synthesis overhead.}
    \label{fig:beyond-1q-emission}
\end{figure}

% 传统上，tableau-based synthesis是在当前搜索步骤中emit任何出现的weight-$1$ Pauli operator (single-qubit Pauli rotation)~\cite{goubault2024faster,yang2025phoenix,kuo2026unified}, 但是严格将weight-$1$ Pauli operator的出现作为terminal case用力过猛，会有反优化的风险。
% 根据 \Cref{the:interaction-rank-cnot}, 任意连续的两个weight-$2$ Pauli operator至多只需要两个CNOT综合，当然单个weight-2 Pauli exponential也需要两个CNOT综合了；更多数目的连续weight-$2$ Pauli operator也未必需要至多三个CNOT才能综合，\Cref{tab:pauli-exp-synth-cost}是一个平均化CNOT综合代价的统计
% \Cref{fig:beyond-1q-emission}中是一个很insightful的案例，说明某些时候在要去搜索Clifford简化的Pauli string达到权重2的时候就应该及时emit，不必继续搜索Clifford将其优化为1。前者很多时候可以致使更少的综合代价，且引入更宽松的emitted block之间的依赖关系，能够潜在支撑更加高效的concurrency scheduling。
% 如果按照逐层直到weight-$1$才elimination的方式~\cite{kuo2026unified,yang2025phoenix,goubault2024faster}，simplification sequence长度也会增加，计算代价也会增加
% Pauli-IR 的代数性质异于binary matrix形式的phase table~\cite{vandaele2022phase,amy2019controlled}和parity matrix~\cite{kissinger2019cnot,wu2019optimization,murphy2023global}，所以不必完全按照phase table or parity matrix的方式去做逐层的weight-$1$ elimination，尤其是对于一些连续的weight-$2$ Pauli operator，直接将其作为terminal case去emit，往往可以获得更低的综合代价和更宽松的依赖关系，从而支持更高效的并行调度。
% in real-world application, 许多只包含2-Local、3-local Pauli terms的optimization问题的Hamiltonian, 更加不适合weight-$1$ as terminal的elimination方式

Existing tableau-based procedures uniformly treat weight-1 as the terminal condition, continuing to simplify a row until it is reduced into a single-qubit rotation before emitting it~\cite{goubault2024faster,yang2025phoenix,kuo2026unified}. We argue that enforcing weight-$1$ as a strict terminal condition is overkill and incurs significant risks of anti-optimization. As illustrated in \Cref{fig:beyond-1q-emission}, relentlessly pushing a target row to weight-$1$ via additional conjugations frequently inflates the weights of unresolved rows, resulting in a net increase in overall synthesis overhead.

A weight-2 row, by contrast, is already efficiently executable. By the interaction-rank criterion of \Cref{thm:interaction-rank-cnot} and the statistical results in \Cref{tab:pauli-exp-synth-cost}, the average synthesis cost of a two-qubit Pauli-evolution block, i.e., an exponential of a sum of two-qubit Pauli operators, $e^{-i\sum_j \theta_j P_j}$, $P_j\in \{X,Y,Z\}^{\otimes 2}$, remains between two and three CNOTs, rather than being invariably three CNOTs as previously assumed. Early emission with weight at most two additionally shortens the Clifford sequence and converts many fine-grained rotations into coarser, independently schedulable blocks, which relaxes the dependencies available to the downstream scheduler.

% Moreover, emitting weight-$2$ blocks directly is algorithmically sound and hardware-efficient. According to the interaction-rank criterion (\Cref{thm:interaction-rank-cnot}), a two-qubit Pauli rotation inherently requires at most two CNOT gates for synthesis. Empirically, as detailed in \Cref{tab:pauli-exp-synth-cost}, continuous emissions of weight-$2$ operators maintain remarkably low average CNOT costs. Early emission of weight-$2$ blocks not only yields lower synthesis overhead but also introduces looser inter-block dependencies compared to dense single-qubit operations, thereby providing downstream concurrency schedulers with a vastly expanded optimization landscape. Finally, relaxing the terminal boundary shortens the overall Clifford simplification sequence and substantially trims compilation time. This is especially advantageous for real-world optimization Hamiltonians heavily populated by 2-local and 3-local Pauli terms, where forcing a weight-$1$ termination is fundamentally counterproductive.
\begin{theorem}[Interaction-rank criterion]\label{thm:interaction-rank-cnot}
Let $U_J=e^{-iH_J}$ be an exact two-qubit Pauli-evolution block with
\begin{equation}
    H_J=\sum\nolimits_{\mu,\nu\in\{X,Y,Z\}}J_{\mu\nu}\,
    \sigma_\mu\otimes\sigma_\nu,
    \quad J\in\mathbb{R}^{3\times 3}.
\end{equation}
If $\operatorname{rank}(J)\leq 2$, then $U_J$ requires at most two CNOT
gates for synthesis. Otherwise, a full-rank $J$ yields a
three-CNOT unitary for generic coefficients.  The interaction rank therefore
provides an analytic two- versus three-CNOT test without numerical KAK
decomposition. \appendixnote{The detailed proof is provided in Appendix~\ref{sec:appendix-interaction-rank}.}
\end{theorem}
\iffalse
\begin{proof}[Proof sketch]
The criterion follows from a proper singular-value decomposition of $J$: local basis changes transform $H_J$ into the Cartan form $s_1XX+s_2YY+s_3ZZ$, so $\operatorname{rank}(J)\leq2$ forces the third Cartan component to vanish\appendixnote{---the at-most-two-CNOT boundary---whereas full rank generically leaves all three components nonzero}.
\end{proof}
\fi
Emitting every weight-2 row is nevertheless too aggressive, and the correct boundary is not fixed. While the tableau remains dense, its rows overlap heavily on the same qubits, so the next UCG has a high chance of simplifying several of them at once; retiring a row early forfeits its share of that collective reduction. Once the tableau becomes sparse, this leverage largely disappears while the collateral weight increases become harder to amortize over the few rows that remain. The emission boundary should therefore track the density of the unresolved tableau---always retiring single-qubit rotations, but releasing weight-2 blocks only once further global reduction has little left to exploit.% \symphony\ realizes this adaptive criterion in \Cref{sec:methodology}.

\section{Methodology}\label{sec:methodology}

\begin{figure*}[tbp]
    \centering
    \includegraphics[width=\linewidth]{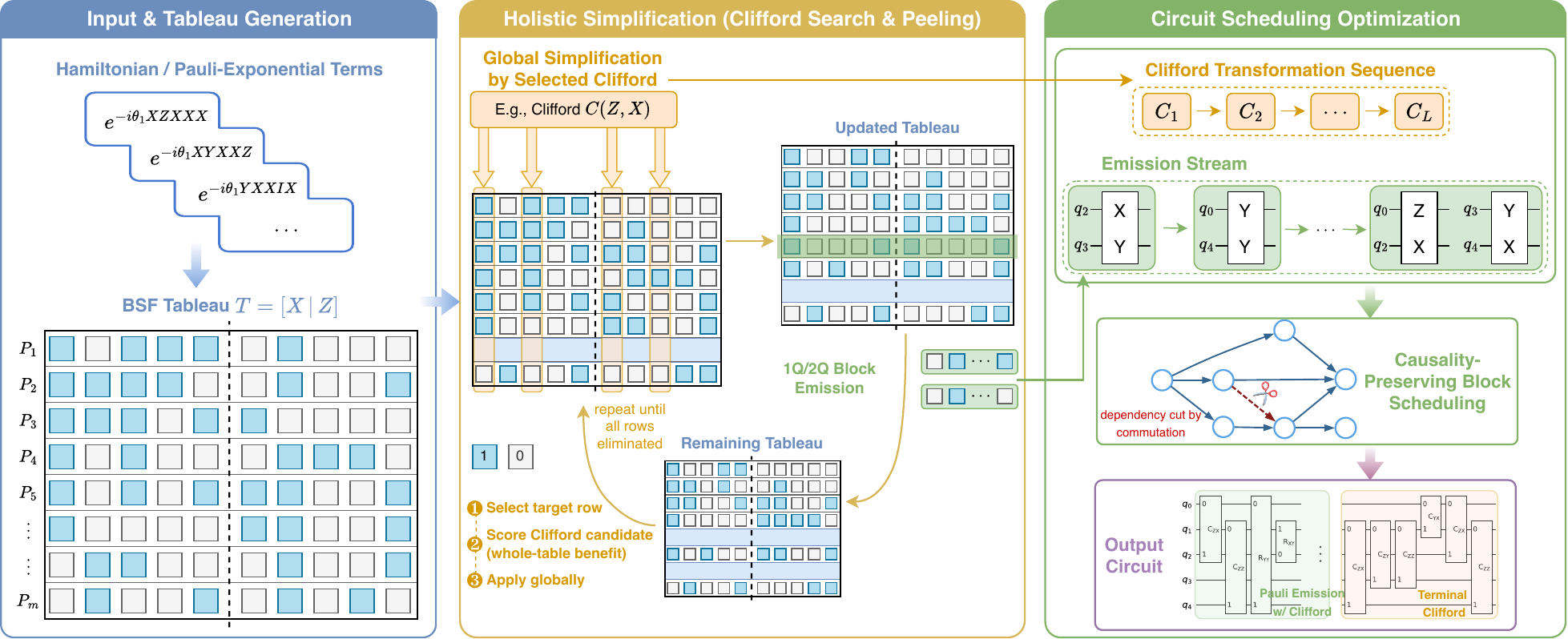}
    % \caption{Overview of the \symphony\ compilation workflow. A global BSF tableau is iteratively simplified via UCG Clifford conjugations guided by a greedy heuristic. Throughout this forward-frame reduction, an adaptive emission mechanism dynamically extracts executable Pauli blocks. Finally, a causality-preserving ASAP scheduling pass exposes extensive two-qubit parallelism across the emitted sequence before a terminal Clifford closes the circuit.}
        % \caption{Overview of the \symphony\ workflow. A unified BSF tableau is simplified by target-reducing UCGs selected for whole-tableau benefit. Eligible rows are emitted in the forward Clifford frame, then causality-preserving block scheduling exposes parallelism before a terminal Clifford closes the frame.}
    \caption{Overview of the \symphony\ compilation workflow.}
    \label{fig:profile}
\end{figure*}

\subsection{Overview}

Guided by the principles established in \Cref{sec:motivation}, we propose \symphony\ for generic  Hamiltonian simulation program compilation. As visualized in \Cref{fig:profile}, \symphony\ maintains all Pauli-IRs in a global BSF tableau and synthesizes them through a forward-frame execution model, peeling off executable operations as the tableau is dynamically simplified.

As formalized in \Cref{algo:holistic-bsf-simp}, the core engine iteratively applies selected UCGs to simplify all active rows within the tableau, emitting eligible weight-$1$ or weight-$2$ Pauli exponentials. This forward-frame simplification process compiles the desired evolution operator into an interleaved sequence of Clifford conjugations $C_\ell$ and Pauli-IR emissions $E_\ell$:
\begin{align}
    \left(C_1\cdots C_L\right)\cdot E_L C_L\cdots E_1 C_1 E_0
    \approx e^{-i\sum_{j=1}^{m}\theta_j P_j},
    \label{eq:forward-frame}
\end{align}
where the rightmost factor $E_0$ is executed first and the bracketed terminal Clifford last. The resulting instruction stream is then mapped onto a dependency graph for a causality-preserving as-soon-as-possible (ASAP) scheduling pass focusing on maximizing two-qubit block parallelism. The terminal Clifford sequence $\prod_{\ell=1}^L C_\ell^\dagger$ can be further resynthesized by standard Clifford optimization, or completely absorbed into the measurement observables by classical post-processing for expectation value computation workloads.

\subsection{Forward-Frame Pauli-Row Reduction and Adaptive Two-Qubit Block Emission}

\symphony\ orchestrates the algebraic reduction of the BSF tableau via a row-by-row target elimination strategy (the outer loop of \Cref{algo:holistic-bsf-simp}). To guarantee forward progress and algorithm termination, the compiler dynamically selects the active row $\tau$ with the lowest non-trivial weight as the immediate simplification target. Since \code{Emit} always extracts rows of weight at most one, a newly selected target satisfies $\operatorname{wt}(P_\tau)\ge 2$. While the target row index $\tau$ is fixed, the algorithm evaluates Clifford candidates with $\Delta_C(P_\tau)=-1$, so the target weight decreases by exactly one per accepted Clifford. Such candidates always exist according to the closed-form UCG conjugation rule of \Cref{sec:appendix_ucg_clifford}.
%  for any pair $\{a,b\}\subseteq \operatorname{supp}(P_\tau)$ carrying letters $(A,B)$, choosing control axis $P_1=A$ and a target axis $P_2$ anticommuting with $B$ annihilates the letter on $a$ and leaves a single non-identity letter on $b$, reducing the row weight by one. 
The target therefore reaches weight one after at most $\operatorname{wt}(P_\tau)-1\le n-1$ Clifford transformations and is then extracted by \code{Emit}.
%   Because $\tau$ is reset only when the target itself is extracted, every row becomes the target at most once.
  
%   By applying strategically chosen UCGs to the overlapping qubits of this target, the row's weight strictly decreases until it becomes eligible for extraction from the global tableau. 

To mitigate the degradation risks of strictly enforcing single-qubit emissions, \symphony\ implements an adaptive two-qubit block emission mechanism (\code{Emit} function in \Cref{algo:holistic-bsf-simp}). The function unconditionally strips out all weight-$1$ Pauli operators, but applies a dynamic density threshold $\rho$ as the executable boundary for weight-$2$ candidates. Specifically, if the average per-qubit Pauli weight of the remaining active rows $\rho_{\mathcal{A}}$, denoted as weight density, falls below the threshold $\rho$ (lines \ref{line:active-qubits}--\ref{line:density-comparison} in \Cref{algo:holistic-bsf-simp}), \symphony\ asserts that the remaining tableau is sparse enough that further global overlapping is unlikely to yield beneficial simultaneous simplifications and therefore emits two-qubit blocks. 

% Rows of weight at most one are always emitted; weight-$2$ rows are emitted only when $\rho_{\mathcal{A}}\leq\rho$, and are otherwise retained so that later UCG candidate sweeps may reduce them with beneficial simplification on other active rows.

% Ideally, in this low-density regime, the algorithm safely emits weight-$2$ Pauli strings as two-qubit block rotations; conversely, if the tableau remains densely packed ($\rho_{\mathcal{A}} > \rho$), weight-$2$ strings are temporarily retained, subjecting them to potential subsequent UCG sweeps that might coincidentally reduce them to weight-$1$ with beneficial simplification on of optimizing other targets.

% Emitted Pauli rotations sharing identical qubit supports are buffered until a subsequent UCG modifies those qubits. These accumulated rotations are then merged and synthesized as a single two-qubit unitary, effectively amortizing the entangling overhead.

\begin{algorithm}[tbp]
    \SetAlgoLined
    \caption{Holistic BSF Simplification}
    \label{algo:holistic-bsf-simp}
    \SetKwInOut{Input}{Input}
    \SetKwInOut{Output}{Output}
    \SetKwBlock{Assumption}{Assumption}{}
    \SetKwFunction{Emit}{Emit}
    \SetKwProg{Fn}{Function}{:}{}

    \Input{Pauli-IRs $\mathcal{P} = \{(P_i, \theta_i) \}_{i=1}^m$ over $n$ qubits, held in a BSF tableau $T=[X\,|\, Z] \in \mathbb{F}_2^{m\times 2 n}$; Clifford options $\mathcal{O}=\bigl\{C{(P,Q)} \mid P,Q\in \{X,Y,Z\}\bigr\}$; weight density threshold $\rho\in[0,1]$}
    \Output{Clifford transformation sequence $\mathcal{C}=\langle C_\ell\rangle_1^L$ and emission stream $\mathcal{E}=\langle E_\ell\rangle_0^L$, realizing $(C_1\cdots C_L)\cdot E_L C_L\cdots E_1 C_1 E_0 \approx e^{-i\sum_{k=1}^{m}\theta_k P_k}$}
    
    \BlankLine
    $\mathcal{A}\gets\{1,\dots,m\}$\tcp*{active rows}
    $\tau\gets\bot$\tcp*{current target row}
    $\mathcal{C},\,\mathcal{E}\gets\varnothing$;\quad \Emit{}\;

    \While{$\mathcal{A}\neq\varnothing$}{
        \lIf{$\tau=\bot$}{$\tau\gets\arg\min_{i\in\mathcal{A}} \textup{wt}(P_i)$}
        % $S\gets\textsc{supp}(P_\tau)$;\quad 
        $\beta^\star\gets-\infty$\;

        \ForEach{$C\in\mathcal{O}$ \textup{and} $\{a,b\}\subseteq \textup{supp}(P_\tau) $}{
            \If{$\Delta_C(P_\tau) := \textup{wt}(C P_\tau C) - \textup{wt}(P_\tau) = -1$}{
                $\Delta W \gets  \sum_{i\in \mathcal{A}} \Delta_C(P_i)$\;\label{line:weight-change}
                $N_{-} \gets \sum_{i\in\mathcal{A}}\mathbf{1}[\Delta_C(P_i) < 0]$\;\label{line:num-decrease}
                $\beta \gets (-\Delta W, N_{-})$\tcp*{Tableau benefit}
                \lIf{$\beta>_{\mathrm{lex}}\beta^\star$}{
                    $\beta^\star\!\gets\!\beta$;\, $(C^\star\!,a^\star\!,b^\star)\!\gets\!(C,a,b)$
                }
            }
        }
        update $T$, $\{P_i\}_{i\in\mathcal{A}}$ by Clifford $(C^\star, a^\star, b^\star)$\;
        append $(C^\star,(a^\star,b^\star))$ to $\mathcal{C}$;\quad \Emit{}\;
    }
    \Return $\mathcal{C},\mathcal{E}$\;

    \BlankLine
    \Fn{\Emit{}}{
        \lIf{$\mathcal{A}=\varnothing$}{\Return}
        $\mathcal{R}\gets\{i\in\mathcal{A}\mid \operatorname{wt}(P_i)\leq 1\}$\;
        \If{$\exists i\in\mathcal A:\operatorname{wt}(P_i)=2$}{
            $Q_{\mathcal A}\gets\bigcup_{i\in\mathcal A}\operatorname{supp}(P_i)$\;\label{line:active-qubits}
            $\rho_{\mathcal A}\gets\frac{1}{|Q_{\mathcal A}|}\frac{1}{|\mathcal A|}\sum_{i\in\mathcal A}\operatorname{wt}(P_i)$\;\label{line:active-density}
            \lIf{$\rho_{\mathcal A}\leq\rho$}{
                $\mathcal{R}\gets\mathcal{R}\cup
                \{i\in\mathcal{A}\mid \operatorname{wt}(P_i)=2\}$\label{line:density-comparison}
            }
        }
        \lIf{$\mathcal{R}=\varnothing$}{\Return}
        append $\big(|\mathcal{C}|,\{P_i\}_{i\in\mathcal{R}},\{\theta_i\}_{i\in\mathcal{R}}\big)$ to $\mathcal{E}$;\quad
        $\mathcal{A}\gets\mathcal{A}\setminus\mathcal{R}$\;
        \lIf{$\tau\in\mathcal{R}$}{$\tau\gets\bot$}
    }
\end{algorithm}

\subsection{Clifford Selection for Simultaneous Weight Reduction of Tableau}

\symphony\ utilizes a lightweight, greedy heuristic that directly quantifies structural weight reduction to strategically select Clifford conjugations. Within an iteration, \Cref{algo:holistic-bsf-simp} evaluates all valid Clifford candidates applied to pairs of the supporting qubits of the currently targeted Pauli row $P_\tau$. Any candidate $C$ that strictly reduces the weight of the target row is scored by its effect on the entire active tableau through two integer quantities: the aggregate weight change $\Delta W$ and the number of rows whose weight decreases $N_{-}$ (defined in line~\ref{line:weight-change} and line~\ref{line:num-decrease}). Candidates are ranked lexicographically by the quantified benefit $\beta=(-\Delta W,\,N_{-})$. This metric explicitly incentivizes simultaneous simplification: it prioritizes the UCG that induces the steepest drop in total tableau weight, breaking ties in favor of the operation that simplifies the highest raw count of Pauli strings. A third component counting rows with increased weight, $N_{+}$, could be incorporated into the benefit metric, but it is actually redundant: because a two-qubit Clifford acting on two support qubits changes each row's weight by at most $\pm 1$, the relation $\Delta W = N_{+} - N_{-}$ holds identically, rendering $N_{+}$ completely determined by $\Delta W$ and $N_{-}$.
\appendixnote{This transparent evaluation equipped with the localized Clifford selection strategy is remarkably robust against the volatility of heterogeneous tableaux, completely eliminating the stalling and oscillation behaviors observed when using convoluted heuristics~\cite{yang2025phoenix}.}

% Evaluating $\beta$ naively would require transforming a copy of the tableau for every candidate. \symphony\ avoids this entirely. A UCG on a qubit pair reads and writes only the four BSF columns of that pair, so each row's behavior is determined by its $4$-bit local pattern, of which there are just $16$. We therefore precompute, for all nine UCGs and all sixteen patterns, the transformed pattern, the induced weight change, and the conjugation sign. For a given pair, one pass over the active rows produces a $16$-bin histogram of their local patterns, after which $\Delta W$ and $N_{-}$ for each of the nine UCGs follow from $16$-dimensional inner products against the precomputed tables. Whole-tableau scoring thus costs one histogram per qubit pair rather than one tableau copy per candidate, which is what makes a global objective affordable at every step.

\subsection{Causality-Preserving Block Scheduling}

% 这段部分的terminology和message尤其要参照我已经在abstract和Introduction已经写的内容来叙述

% 另外，以下两个proposition搭配以往熟知的Pauli operator之间的对易性（也是相同overlapping qubit上的Pauli operator要相同）, Pauli-IR block <-> Pauli-IR block之间、Pauli-IR block <-> UCG Clifford 之间，UCG Clifford <-> UCG Clifford之间的对易性判定，都是可以通过local Pauli letter的比较来实现的。所以这也就给我我们canonical commutation relations，方便进一步做scheduling上的优化

Because \symphony\ peels Pauli-IRs dynamically within a forward Clifford frame, the emitted sequence inherently carries structural precedences (e.g., a Pauli block $E_1$ must logically execute after the Clifford $C_1$ that exposed it). In this instruction stream, rotations emitted on the same qubit support are merged into a single Pauli block until a UCG acts on any of these qubits (\Cref{fig:block-scheduling}(a)). However, rigidly executing this sequence sequentially leaves substantial parallelism unexploited. To aggressively optimize circuit depth, we introduce a causality-preserving ASAP block rescheduling phase, implemented via a graph edge-coloring heuristic.

\begin{figure*}[tbp]
    \centering
    \includegraphics[width=0.85\linewidth]{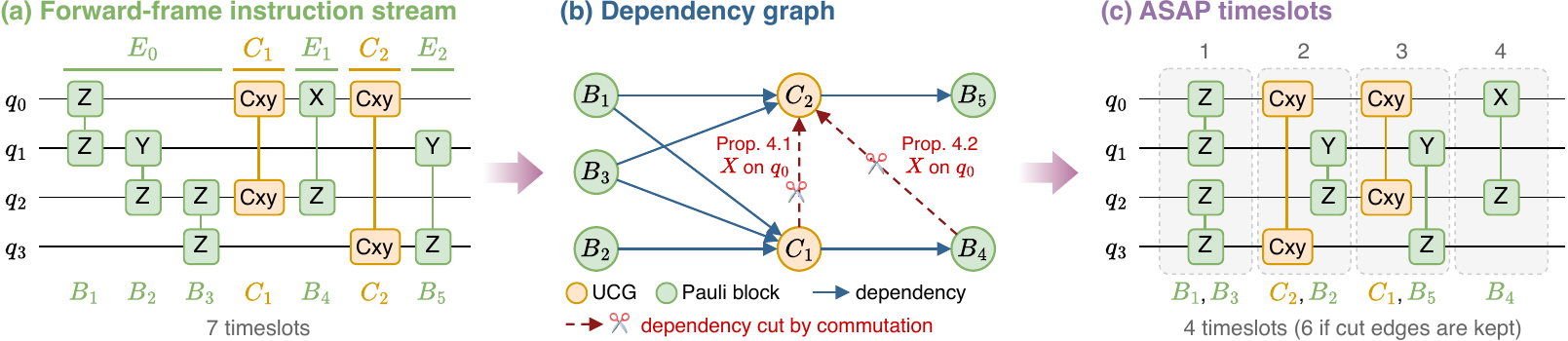}
    \caption{Causality-preserving block scheduling on a four-qubit example. Green boxes give the Pauli letters of each block; \textsf{Cxy} denotes the UCG $C(X,Y)$ with $X$ on the upper qubit.}
    \label{fig:block-scheduling}
\end{figure*}

This scheduling engine also exploits canonical commutation relations alongside inherent sequence flexibility to pull operations forward without violating logical equivalences. Since Pauli blocks and UCGs natively expose their algebraic structures, \symphony\ resolves broad commutativity checks analytically via local Pauli letter comparisons. Let $M_q \in \{I, X, Y, Z\}$ denote the Pauli letter of operator $M$ on qubit $q$. Using the projection form $C_{a,b}(P,Q)=I-2\Pi(P)_a\Pi(Q)_b$ with $\Pi(P)=(I-P)/2$, we establish two foundational rules:

\begin{proposition}[UCG--UCG commutation]\label{prop:ucg-ucg-commutation}
Two UCGs commute if and only if they carry identical Pauli letters on all shared qubits. Thus, commutativity is trivially satisfied for disjoint operations or upon matching intersecting Paulis.
\end{proposition}

\begin{proposition}[UCG--Pauli-rotation commutation]\label{prop:ucg-pauli-commutation}
$R_M(\theta)=e^{-i\theta M}$ commutes with $C_{a,b}(P,Q)$ if and only if $M_a\in\{I,P\}$ and $M_b\in\{I,Q\}$. Equivalently, the Pauli rotation must match the UCG's Pauli letters on all overlapping qubits.
\end{proposition}

% \begin{proposition}[UCG--UCG commutation]\label{prop:ucg-ucg-commutation}
% Two UCG operations commute if and only if they carry the exact same Pauli letter on every qubit they share. Thus, disjoint UCGs always commute, two UCGs on the identical qubit pair commute if both Pauli letter match, and two UCGs sharing one qubit commute if their Pauli letters on that intersecting qubit agrees.
% \end{proposition}

% \begin{proposition}[UCG--Pauli-rotation commutation]\label{prop:ucg-pauli-commutation}
% A nontrivial Pauli rotation $R_M(\theta)=e^{-i\theta M}$ commutes with a UCG $C_{a,b}(P,Q)$ if and only if $M_a\in\{I,P\}$ and $M_b\in\{I,Q\}$. Equivalently, the Pauli rotation must carry the UCG's Pauli letter on every overlapping qubit; its letters on non-overlapping qubits are irrelevant.
% \end{proposition}

As illustrated in \Cref{fig:block-scheduling}, \symphony\ constructs an DAG by evaluating the interleaved sequences $\mathcal{E}$ and $\mathcal{C}$ against these commutativity criteria, linking each operation to every earlier non-commuting operation on a shared qubit. Pauli blocks are never linked to each other, since reordering Trotterized exponentials yields an equally valid product formula. An edge-coloring-inspired heuristic subsequently assigns dependency-ready, resource-disjoint operations to the same timeslot, placing each in the earliest timeslot its dependencies and qubits allow, thereby preserving frame-induced causal precedences while exploiting extensive two-qubit block parallelism from the serial emission stream. In \Cref{fig:block-scheduling}, $C_1\to C_2$ and $B_4\to C_2$ are cut because both ends carry $X$ on $q_0$ (\Cref{prop:ucg-ucg-commutation,prop:ucg-pauli-commutation}), so $C_2$ advances from the fifth to the second timeslot and the schedule shrinks from six timeslots under frame-induced precedences alone to four.

\subsection{Complexity Analysis}

% \symphony\ exhibit the asymptotic complexity of $\mathcal{O}(m^2 n^3)$ with $m$ Pauli strings and $n$ qubits.
Within each iteration, \Cref{algo:holistic-bsf-simp} evaluates a bounded $\mathcal{O}(n^2)$ search space of UCG Clifford candidates against the target row $P_\tau$. Computing the global tableau weight differential costs $\mathcal{O}(m)$ per candidate, resulting in an $\mathcal{O}(mn^2)$ time complexity to discover the optimal simplifying Clifford per step. Because each chosen Clifford strictly reduces the target row's weight, the maximum number of iterations is bounded by the tableau's total initial weight, $\mathcal{O}(mn)$. Consequently, the complete tableau simplification phase operates 
\appendixnote{within a highly tractable $\mathcal{O}(m^2n^3)$ upper-bound complexity.}
% within an $\mathcal{O}(m^2n^3)$ complexity.
\iffalse
The scheduling phase handles $K\le L+m$ operations, where $L=\mathcal{O}(mn)$ is the number of applied UCGs and each emitted rotation joins exactly one Pauli block. Since commutation is not transitive, building the dependency graph tests each operation against all earlier operations on its qubits, i.e., $\mathcal{O}(K^2)$ tests, each resolved in constant time by comparing cached per-qubit Pauli letters (\Cref{prop:ucg-ucg-commutation,prop:ucg-pauli-commutation}); greedy timeslot assignment over the graph adds at most $\mathcal{O}(K^2)$. Scheduling thus costs $\mathcal{O}(m^2n^2)$ and leaves the overall $\mathcal{O}(m^2n^3)$ bound unchanged, which keeps the compilation cost within the same polynomial bound.
\fi
The scheduling phase processes $K\le L+m=\mathcal{O}(mn)$ operations, where $L$ is the number of applied UCGs. Because commutation is not transitive, constructing the dependency graph requires $\mathcal{O}(K^2)$ pairwise checks, each resolved in constant time (\Cref{prop:ucg-ucg-commutation,prop:ucg-pauli-commutation}); greedy timeslot assignment adds at most another $\mathcal{O}(K^2)$. Scheduling therefore costs $\mathcal{O}(m^2n^2)$ and leaves the overall $\mathcal{O}(m^2n^3)$ bound unchanged.

\section{Evaluation}\label{sec:evaluation}

We evaluate \symphony\ against prior SOTA compilers on a diverse set of Hamiltonian simulation programs. \symphony\ is implemented in Python on top of Qiskit and NumPy. Code and data are available on GitHub~\cite{phoenixGitHub}.
% We conduct experiments on MacOS with Apple M3 MAX CPU and 36GB RAM.
We conduct experiments on a dual-socket server equipped with two 96-core AMD EPYC 9684X processors and 1.5 TiB of RAM.

\subsection{Experimental Settings}

% \begin{table*}[tbp]
%     \centering
%     \caption{Summary of the HamLib benchmark suite.}
%     \label{tab:hamlib-benchmark-description}
%     % \setlength{\tabcolsep}{4pt}
%     \small\input{tables/hamlib_suite_summary.tex}    
% \end{table*}

\begin{table}[tbp]
    \centering
    \caption{Summary of the evaluated benchmarks. (a) The HamLib suite comprises 100 Hamiltonians across four application categories: Binary Optimization (15), Discrete Optimization (15), Chemistry (35), and Condensed Matter (35). (b) The UCCSD suite contains six synthetic Hamiltonians.}

    \subfloat[Summary of the HamLib benchmark suite\label{tab:hamlib-benchmark-description}]{
    \begin{footnotesize}
        % \begin{tabular}{|l|c|c|c|c|c|c|}
% \hline
% Category & \#Hamiltonian & \#Qubit & \#Pauli & Weight & 2Q gate count & 2Q circuit depth \\
% \hline
% Binary optimization & 15 & 4--90 & 12--2958 & 2--3 & 24--5916 & 24--4268 \\
% \hline
% Discrete optimization & 15 & 8--480 & 48--5334 & 2--8 & 128--24756 & 44--22272 \\
% \hline
% Chemistry & 35 & 2--64 & 6--6508 & 2--24 & 8--137080 & 4--125891 \\
% \hline
% Condensed matter & 35 & 12--930 & 32--8241 & 2--361 & 32--42768 & 28--38201 \\
% \hline
% \emph{All} & 100 & 2--930 & 6--8241 & 2--361 & 8--137080 & 4--125891 \\
% \hline
% \end{tabular}

\begin{tabular}{|l|c|c|c|c|c|c|}
\hline
Category (\#) & \#Qubit & \#Pauli & Wt. & \makecell{2Q gate\\count} & \makecell{2Q circuit\\depth} \\
\hline
Binary (15) & 4--90 & 12--2958 & 2--3 & 24--5916 & 24--4268 \\
\hline
Discrete (15) & 8--480 & 48--5334 & 2--8 & 128--24756 & 44--22272 \\
\hline
Chem. (35) & 2--64 & 6--6508 & 2--24 & 8--137080 & 4--125891 \\
\hline
Cond. (35) & 12--930 & 32--8241 & 2--361 & 32--42768 & 28--38201 \\
\hline
\emph{All} (100) & 2--930 & 6--8241 & 2--361 & 8--137080 & 4--125891 \\
\hline
\end{tabular}

    \end{footnotesize}
    }

    \subfloat[Summary of the UCCSD benchmark suite\label{tab:uccsd-benchmark-description}]{
    \begin{footnotesize}
        \begin{tabular}{|l|c|c|c|c|c|}
\hline
Bench. & \#Qubit & \#Pauli & Wt. & \makecell{2Q gate\\count} & \makecell{2Q circuit\\depth} \\
\hline
UCC-10 & 10 & 800 & 4--7 & 7328 & 7273 \\
\hline
UCC-15 & 15 & 1800 & 4--15 & 24720 & 24198 \\
\hline
UCC-20 & 20 & 3200 & 4--20 & 57664 & 55729 \\
\hline
UCC-25 & 25 & 5000 & 4--25 & 115888 & 112104 \\
\hline
UCC-30 & 30 & 7200 & 4--30 & 189840 & 182722 \\
\hline
UCC-35 & 35 & 9800 & 4--35 & 301072 & 289930 \\
\hline
\end{tabular}

    \end{footnotesize}
    }
\end{table}

\subsubsection{Baselines}

We compare \symphony\ against six leading baseline compilers: \ding{172} \qiskit-\rustiq\ uses Qiskit's built-in implementation of Rustiq~\cite{goubault2024faster}; \ding{173} \tket-\pcoast\ uses the Pauli Frame Graph synthesis method from \citet{schmitz2024graph} combined with final Clifford resynthesis from PCOAST~\cite{paykin2023pcoast}, specifically employing the built-in \code{GreedyPauliSimp} pass that significantly outperforms the conventional \code{PauliSimp} pass~\cite{cowtan2019phase} commonly used in prior work; \ding{174} \paulihedral~\cite{li2022paulihedral} and \ding{175} \tetris~\cite{jin2024tetris} are two representative methods that exploit opportunities for gate cancellation and routing-synthesis co-optimization across different Pauli-IR synthesis variants; \ding{176} \quclear~\cite{liu2025quclear} extracts intrinsic Clifford transformations to maximize Pauli-IR commutation and CNOT-tree synthesis efficiency, though we explicitly retain and synthesize its terminal Clifford operations---rather than absorbing them into measurement observables---to ensure all methods are compared fairly as executable unitary circuits with identical output semantics; \ding{177} \phoenix~\cite{yang2025phoenix} adopts a support-grouped BSF simplification and inter-group ordering strategy.

For \paulihedral\ and \tetris, we extend their implementations to instantiate circuits with actual rotation angles of input programs.
% the implementations bundled with the \tetris\ artifact, extend both to emit the actual rotation angles of the input program, and fix a bug in \paulihedral's superconducting backend, inherited from its original artifact, that synthesizes incorrect Pauli strings whenever the qubits of a block form a cycle on the coupling graph. 
Following its own artifact, \paulihedral\ uses the default depth-oriented block scheduling in our benchmarking, with its fault-tolerant backend under all-to-all connectivity and its superconducting backend otherwise. 
\iffalse
... its superconducting backend otherwise; its alternative gate-count-oriented scheduling uses $11\%$ fewer two-qubit gates on HamLib but yields $3.05\times$ deeper circuits.
\fi

% Both baselines and \symphony\ are appended by Qiskit's O3 circuit-level optimization pass with two-qubit peephole optimal Clifford subcircuit resynthesis~\cite{bravyi2021hadamardfree}.

All circuits produced by these high-level compilers are further processed by Qiskit's circuit-level optimization, including local cancellation, optimal three-qubit Clifford resynthesis~\cite{bravyi2021hadamardfree}, and two-qubit block consolidation and CNOT-optimal decomposition. In the limited-connectivity compilation experiments, the target coupling map is supplied to the topology-aware \paulihedral\ and \tetris; all outputs are ultimately mapped to the same coupling map by Qiskit's topology-aware O3 \code{transpile} pass, which incorporates SABRE-based qubit layout and routing~\cite{li2019tackling,zou2024lightsabre}.

% \note{The external baselines span Pauli-network synthesis, block-level Pauli-IR synthesis, and Clifford-extraction-based optimization}

\subsubsection{Benchmarks}

Our primary evaluation uses 100 representative Hamiltonian simulation programs from Benchpress~\cite{nation2025benchmarking}, recently featured in Qiskit's performance benchmarking. Originally sourced from the HamLib library~\cite{sawaya2024hamlib}, these Hamiltonians span diverse domains, including quantum chemistry, condensed matter physics, discrete optimization, and binary optimization. This diverse suite encapsulates a wide range of Hamiltonian types and complexities (fermionic electronic-structure and Fermi-Hubbard instances, bosonic/vibrational and Bose-Hubbard instances, spin models such as the transverse-field Ising and Heisenberg models, etc.), varying qubit numbers and interaction terms, with detailed characteristics summarized in \Cref{tab:hamlib-benchmark-description}. We further adopt a supplementary set of six synthetic benchmarks from variational UCCSD ans\"atze to assess scalability on large-scale instances. As listed in \Cref{tab:uccsd-benchmark-description}, they are denoted by UCC-10 through UCC-35, generated by randomly sampling $n^2$ blocks from the original molecular UCCSD ans\"atz constructed by PySCF, which are adopted from the \tetris\ evaluation~\cite{jin2024tetris}.

Notably, the evaluated compilers exploit Pauli term ordering freedom at different levels, which is another reason we adopt both HamLib and UCCSD benchmark sets. Specifically, \symphony\ and \qiskit-\rustiq\ use string-level Pauli reordering; \phoenix\ uses block-level reordering by grouping same-support strings in each block. On UCCSD, \tetris\ always receives one block per excitation, as does \paulihedral\ on limited-connectivity topologies; on all-to-all topologies, \paulihedral's FT backend mode synthesizes only single-string blocks, so it reorders individual strings. On HamLib, \paulihedral\ receives one string per block, as in its own Hamiltonian benchmarking, i.e., the same string-level freedom as \symphony; \tetris\ constructs blocks as \phoenix\ does, as string-level ordering otherwise performs worse.
On small-scale instances, all evaluated compilers produce circuits whose Trotter error, as measured by unitary infidelity, is of the same order of magnitude (e.g., $10^{-5}$), which also provides evidence that the compiled circuits are correct. Notably, these same-order-of-magnitude algorithmic errors induced by various Pauli term ordering strategies are negligible relative to physical noise, while different orderings offer varying opportunities for optimizing gate count and circuit depth, thereby dominating the accuracy of Hamiltonian simulations on near-term and early-fault-tolerant hardware.

\subsubsection{Metrics}
We mainly report the two-qubit gate count and circuit depth, since single-qubit gates are generally of constant overhead relative to two-qubit gates and the circuit-level $\{\mathrm{U3},\mathrm{CNOT}\}$ representation does not reflect hardware primitives in the NISQ regime. The na\"ively synthesized circuits through CNOT trees serve as the reference for reported optimization rates in terms of gate count or circuit depth. For early fault-tolerant benchmarking, the $T$ gate or $R_Z$ gate count and depth are more critical metrics than inexpensive Clifford gates. For cross-compiler comparisons, we highlight the geometric mean of per-program ratios between each baseline and \symphony\ over the common instance subset.

\subsection{Main Results}

\begin{table}[t]
    \caption{Geometric mean of optimization rates on the HamLib suite (QuCLEAR timed out on 4 of the 100 programs).}
    \label{tab:hamlib}
    \centering
    \setlength{\tabcolsep}{3.4pt}
    {%
    \captionsetup[subfloat]{position=top}%
    \subfloat[Average optimization rate in terms of 2Q gate count\label{tab:hamlib-count}]{%
\begin{footnotesize}
\begin{tabular}{|c|c|c|c|c|c|c|c|}
\hline
Category (\#) & Qiskit & TKET & PH. & Tetris & QuC. & Phoenix & Symp. \\
\hline
Binary (15) & 1.252 & 0.886 & 0.91 & \textbf{0.715} & 1.598 & 0.718 & 0.732 \\
\hline
Discrete (15) & \textbf{0.542} & 0.556 & 0.817 & 0.735 & 0.588 & 0.753 & 0.688 \\
\hline
Chem. (35) & 0.292 & 0.276 & 0.458 & 0.457 & 0.354 & 0.327 & \textbf{0.226} \\
\hline
Cond. (35) & 1.084 & 0.855 & 0.725 & 0.627 & 0.83 & 0.493 & \textbf{0.467} \\
\hline
\emph{All (100)} & \emph{0.63} & \emph{0.542} & \emph{0.65} & \emph{0.586} & \emph{0.639} & \emph{0.482} & \emph{\textbf{0.411}} \\
\hline
\end{tabular}%
\end{footnotesize}%
}

\subfloat[Average optimization rate in terms of 2Q circuit depth\label{tab:hamlib-depth}]{%
\begin{footnotesize}
\begin{tabular}{|c|c|c|c|c|c|c|c|}
\hline
Category (\#) & Qiskit & TKET & PH. & Tetris & QuC. & Phoenix & Symp. \\
\hline
Binary (15) & 1.277 & 0.626 & 0.523 & 0.527 & 1.603 & 0.261 & \textbf{0.206} \\
\hline
Discrete (15) & 0.595 & 0.35 & 0.397 & 1.629 & 0.641 & 0.377 & \textbf{0.179} \\
\hline
Chem. (35) & 0.218 & 0.182 & 0.333 & 0.358 & 0.273 & 0.236 & \textbf{0.134} \\
\hline
Cond. (35) & 0.357 & 0.799 & 0.075 & 0.113 & 0.724 & 0.069 & \textbf{0.031} \\
\hline
\emph{All (100)} & \emph{0.393} & \emph{0.405} & \emph{0.217} & \emph{0.318} & \emph{0.563} & \emph{0.167} & \emph{\textbf{0.09}} \\
\hline
\end{tabular}%
\end{footnotesize}%
}

    }
\end{table}

\begin{figure*}[tbp]
    \centering
    \subfloat[\qiskit-\rustiq\ vs. \symphony]{
        \includegraphics[width=0.315\linewidth]{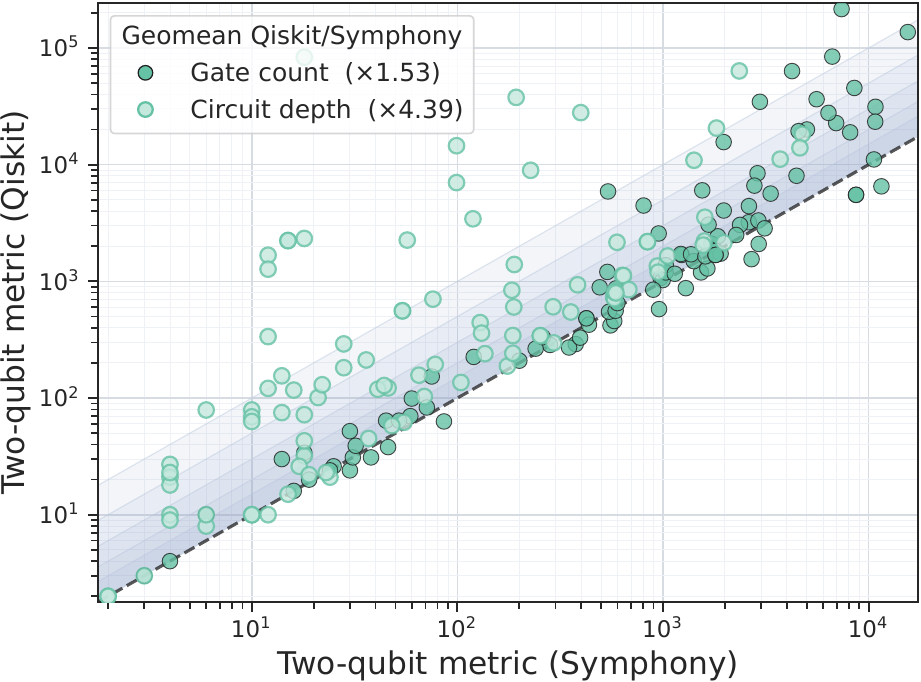}
        \label{fig:hamlib-qiskit-vs-symphony}
    }\hfil
    \subfloat[\tket-\pcoast\ vs. \symphony]{
        \includegraphics[width=0.315\linewidth]{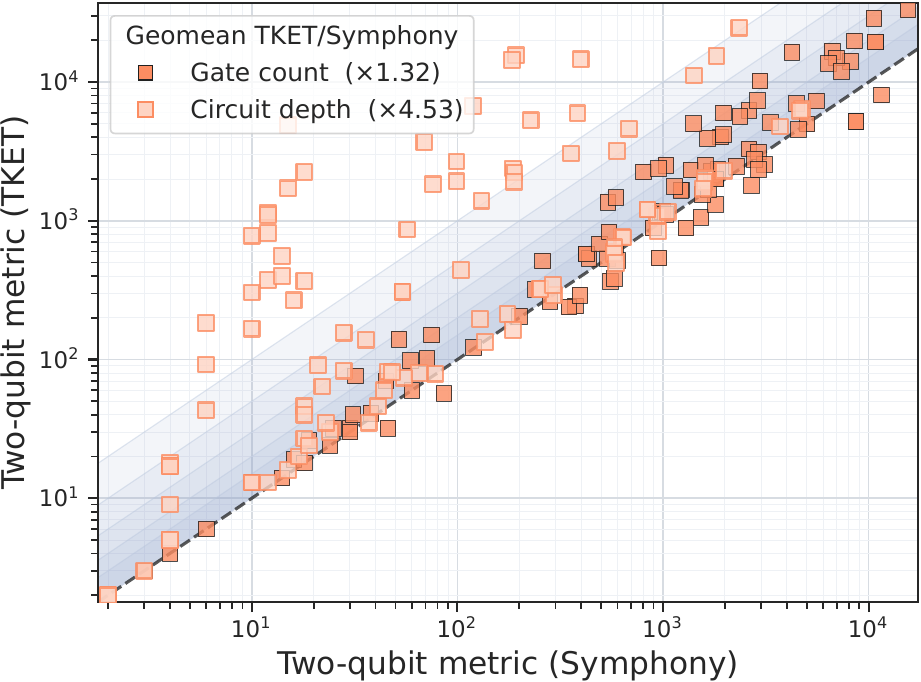}
        \label{fig:hamlib-tket-vs-symphony}
    }\hfil
    \subfloat[\paulihedral\ vs. \symphony]{
        \includegraphics[width=0.315\linewidth]{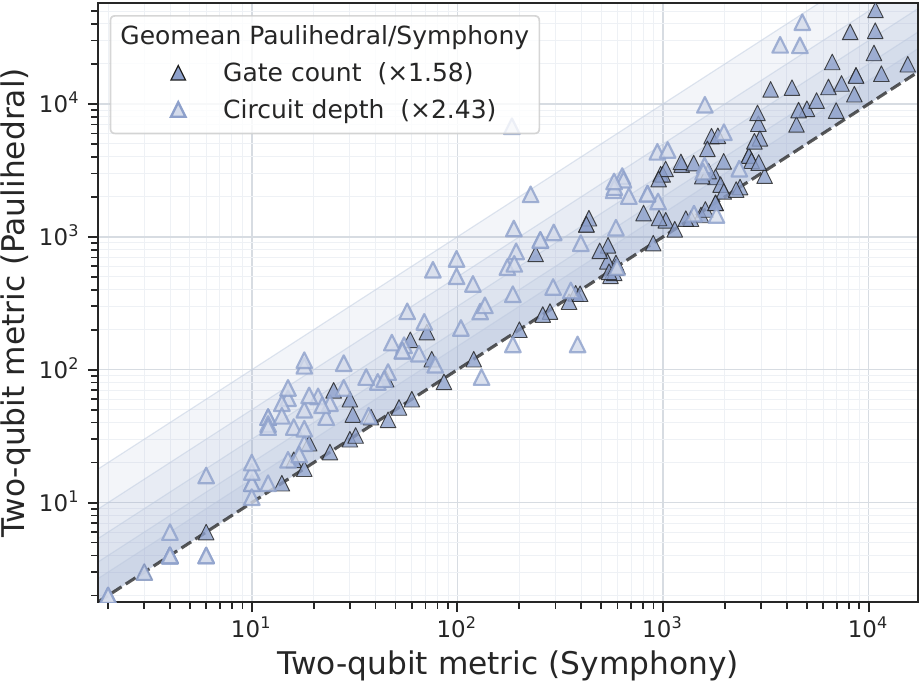}
        \label{fig:hamlib-paulihedral-vs-symphony}
    }\hfil
    \subfloat[\tetris\ vs. \symphony]{
        \includegraphics[width=0.315\linewidth]{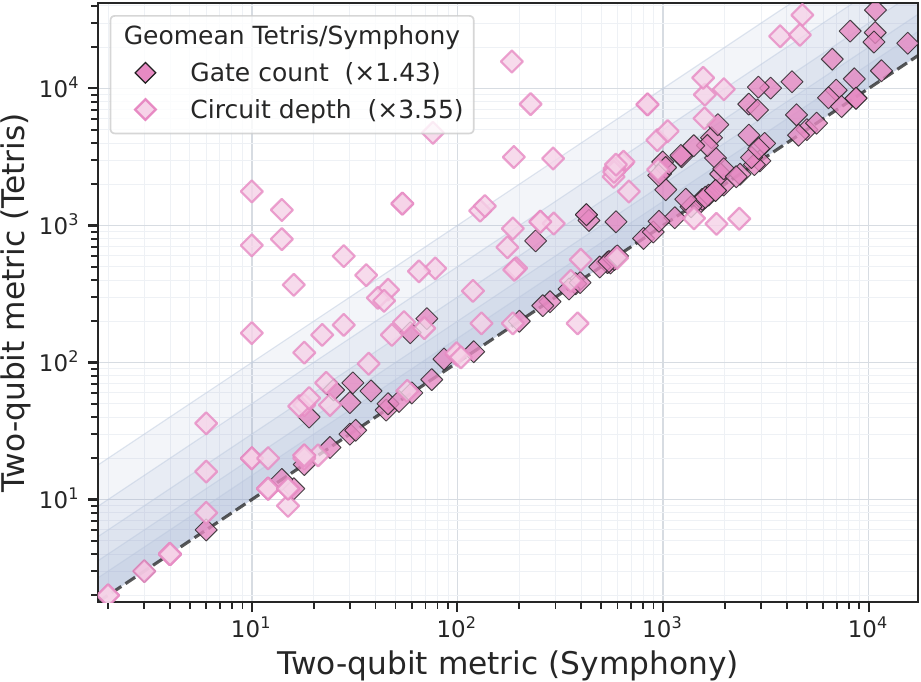}
        \label{fig:hamlib-tetris-vs-symphony}
    }\hfil
    \subfloat[\quclear\ vs. \symphony]{
        \includegraphics[width=0.315\linewidth]{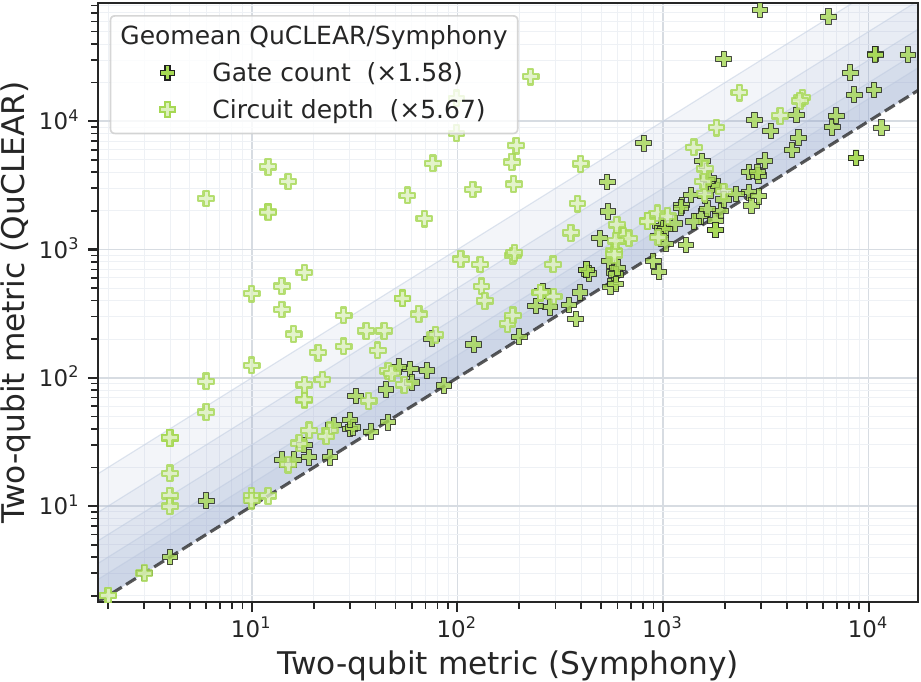}
        \label{fig:hamlib-quclear-vs-symphony}
    }\hfil
    \subfloat[\phoenix\ vs. \symphony]{
        \includegraphics[width=0.315\linewidth]{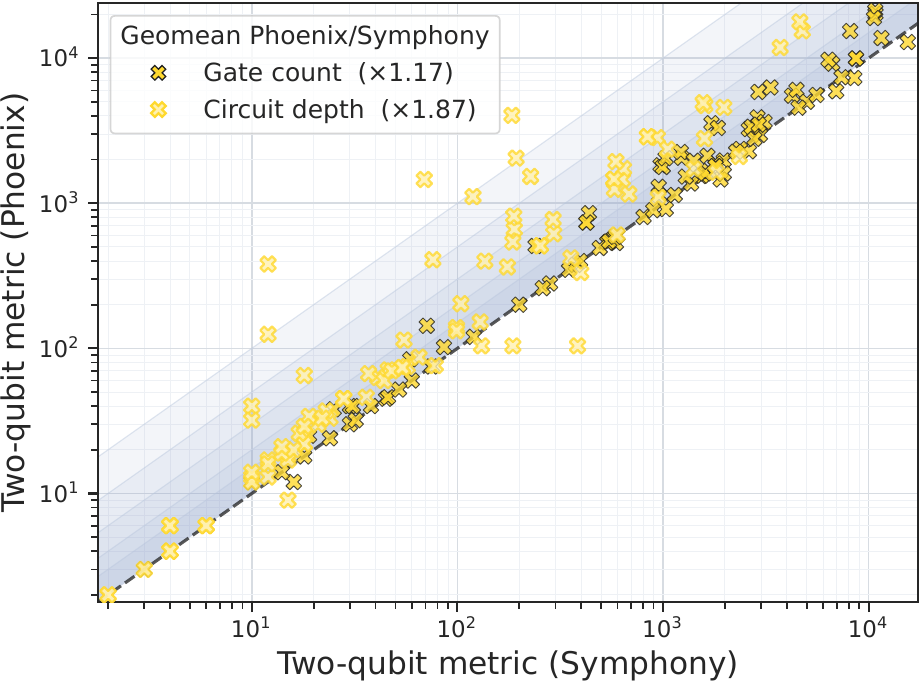}
        \label{fig:hamlib-phoenix-vs-symphony}
    }\hfil
    \caption{Pairwise comparison of \symphony\ against six SOTA baselines on the 100-program HamLib benchmark suite. The dashed diagonal line denotes equal performance; points above it favor \symphony; shaded diagonal bands serve as multiplicative-ratio guides (1--1.5$\times$, 1.5--2$\times$, 2--3$\times$, 3--5$\times$, and 5--10$\times$, respectively). Both gate count and circuit depth metrics count only two-qubit gates.}
    \label{fig:hamlib}
\end{figure*}

\Cref{tab:hamlib} and \Cref{fig:hamlib} summarize the main results on the 100 Hamiltonians from HamLib. \symphony\ attains geometric-mean optimization rates of $0.411$ for two-qubit gate count and $0.09$ for two-qubit depth, corresponding to reductions of $59\%$ and $91\%$, respectively, relative to na\"ive per-term synthesis. Compared to all baselines, \symphony\ achieves the best aggregate optimization rates for both metrics, requiring $1.17$--$1.58\times$ fewer two-qubit gates and $1.87$--$5.67\times$ shallower two-qubit depth, where the ratios are geometric means of per-program baseline-to-\symphony\ ratios on the common instance subsets. 
% Besides, \quclear\ timeout on 4 benchmarks, respectively. 
Remarkably, compared to \phoenix, which follows a similar compilation paradigm but uses grouping-then-ordering, the aggregate gate-count and depth optimization rates improve from $0.482$ to $0.411$ and from $0.167$ to $0.09$, respectively. This substantial margin empirically validates the necessity and power of grouping-free, holistic synthesis.

A breakdown by benchmark family reveals that \symphony's superiority in circuit depth is exceptionally consistent. It achieves the best two-qubit depth optimization rate in every single category: 0.206 (binary optimization), 0.179 (discrete optimization), 0.134 (chemistry), and 0.031 (condensed matter). It is most pronounced on condensed-matter instances, where \symphony's $0.031$ optimization rate corresponds to per-program geometric-mean depth ratios of $2.22\times$ versus the next-best \phoenix\ result and $2.41\times$ versus \paulihedral. The pairwise plots in \Cref{fig:hamlib} further show that these gains are distributed across individual programs rather than driven by a small number of outliers. We acknowledge that \symphony\ does not achieve the lowest gate count in every sub-category; \tetris\ and \phoenix\ retain marginal count advantages on binary optimization ($0.715$ and $0.718$ versus $0.732$), and \qiskit-\rustiq\ and \tket-\pcoast\ require fewer two-qubit gates on discrete optimization ($0.542$ and $0.556$ versus $0.688$). Nevertheless, \symphony's universal supremacy in circuit depth, coupled with its best-in-class aggregate gate count across the full suite, firmly establishes its Pareto dominance.

\subsection{Performance on Limited-Connectivity Backends}

\begin{table}[t]
    \caption{Compilation results on UCCSD programs across both all-to-all and limited-connectivity backends.}
    \label{tab:uccsd}
    \centering
    \setlength{\tabcolsep}{3.8pt}
    {%
    \captionsetup[subfloat]{position=top}%
    \subfloat[Average optimization rate in terms of 2Q gate count\label{tab:uccsd-count}]{%
\begin{footnotesize}
\begin{tabular}{|c|c|c|c|c|c|c|c|}
\hline
Topo. & Qiskit & TKET & PH. & Tetris & QuC. & Phoenix & Symphony \\
\hline
All2all & 0.256 & 0.175 & 0.333 & 0.3 & 0.228 & 0.133 & \textbf{0.131} \\
\hline
Square & 1.117 & 0.691 & 0.65 & 0.607 & 0.952 & \textbf{0.395} & 0.452 \\
\hline
HHex & 1.95 & 1.156 & 0.807 & 0.718 & 1.614 & \textbf{0.567} & 0.757 \\
\hline
\end{tabular}%
\end{footnotesize}%
}

\subfloat[Average optimization rate in terms of 2Q circuit depth\label{tab:uccsd-depth}]{%
\begin{footnotesize}
\begin{tabular}{|c|c|c|c|c|c|c|c|}
\hline
Topo. & Qiskit & TKET & PH. & Tetris & QuC. & Phoenix & Symphony \\
\hline
All2all & 0.153 & 0.059 & 0.283 & 0.299 & 0.103 & 0.112 & \textbf{0.051} \\
\hline
Square & 0.575 & 0.309 & 0.497 & 0.434 & 0.456 & 0.276 & \textbf{0.21} \\
\hline
HHex & 0.911 & 0.466 & 0.627 & 0.511 & 0.705 & 0.362 & \textbf{0.319} \\
\hline
\end{tabular}%
\end{footnotesize}%
}

    }
\end{table}

We further evaluate whether \symphony's logical-level compilation advantage still holds when applied to restricted hardware topologies by using the UCCSD benchmark suite, whose regular scaling enables a controlled comparison of routing overhead. While \paulihedral\ and \tetris\ incorporate backend topology information and perform co-optimization, \symphony\ and other topology-agnostic compilers assume all-to-all connectivity and delegate qubit mapping to Qiskit.

As \Cref{tab:uccsd} details, \symphony\ secures the best two-qubit depth optimization rate across all topologies---$0.051$ (all-to-all), $0.21$ (square), $0.319$ (heavy-hex). Topology-aware baselines (\paulihedral\ and \tetris) demonstrate strong resilience to restricted connectivity; transitioning from all-to-all to heavy-hex introduces minor depth overheads of 2.22$\times$ and 1.71$\times$, respectively. % (calculated as $0.627/0.283$) and $1.71\times$ ($0.511/0.299$), respectively.
In contrast, \symphony\ experiences a steeper $6.25\times$ (calculated as $0.319/0.051$) routing overhead. However, despite this faster degradation, \paulihedral\ and \tetris\ still trail \symphony\ on the restrictive heavy-hex lattice by $1.97\times$ and $1.60\times$, respectively, as their logical-level compiled circuits are substantially inferior to those of \symphony.
% Ultimately, robustness to the coupling map cannot recover the optimization opportunities surrendered before routing begins.
Regarding two-qubit gate count, \symphony\ significantly outperforms all baselines on the all-to-all topology. On restricted topologies, however, it trails \phoenix\ by 14--34\% (and \tetris\ by 5\% on heavy-hex) due to larger degradation (e.g., $5.78\times$ versus $4.26\times$ gate count overhead on heavy-hex relative to all-to-all). This highlights a fundamental architectural trade-off: aggressive holistic BSF simplification dissolves the same-support groups preserved by \phoenix, depriving the backend router of exploitable hardware-friendly program patterns.

These results challenge the common premise that introducing connectivity constraints early necessarily improves mapped circuit quality~\cite{li2022paulihedral,jin2024tetris,li2025pauliforest}. We attribute the superiority of hardware-agnostic compilers here to the fact that unrestricted algebraic simplification over logical-level Pauli-IRs yields optimization gains that outweigh subsequent routing overhead, particularly for Hamiltonians such as chemistry simulation generated through fermion-to-qubit mappings with diverse non-local Pauli operators acting on qubit sets. With modern routing algorithms continually advancing, high-level engines can increasingly outsource spatial mapping to maximize purely algebraic compression. However, this separation introduces a distinct trade-off: aggressively dissolving native program patterns complicates qubit routing, causing steeper degradation on limited-connectivity hardware compared to topology-aware or grouping-based baselines like \paulihedral, \tetris, and \phoenix. Furthermore, for specific Hamiltonians exhibiting strong structural patterns and high locality—such as those in condensed matter simulations—topology-aware co-optimization remains highly beneficial~\cite{kattemolle2026efficient}. Therefore, the co-design of layout synthesis and aggressive algebraic optimization is a highly promising avenue for future exploration.

\subsection{Superiority in the Early Fault-Tolerant Regime}

\begin{figure}[tbp]
    \centering
    \includegraphics[width=\columnwidth]{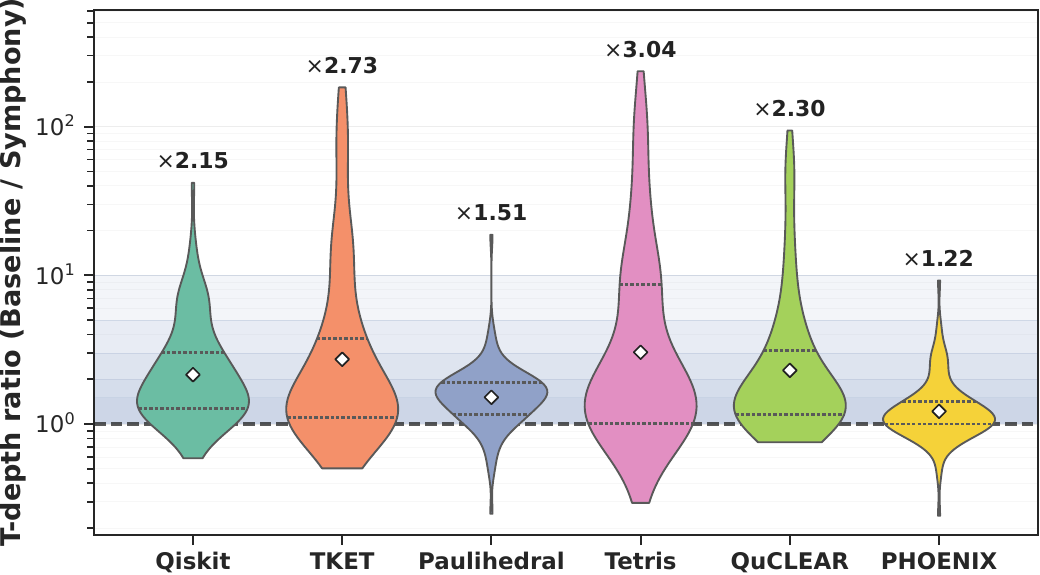}
    \caption{Distribution of $T$-depth ratios of each baseline to \symphony\ across HamLib, annotated with geometric means. The dashed line represents equal $T$-depth; ratios above it favor \symphony; shaded regions indicate multiplicative-ratio bounds (1--1.5$\times$, 1.5--2$\times$, 2--3$\times$, 3--5$\times$, and 5--10$\times$).}
    %  \quclear\ timeouts on xxx out of 100 programs.}
    % Legend entries specify the geometric mean of the per-program ratio relative 
    % to \symphony. Note that points are paired per program; thus, the \quclear\ series 
    % reflects the xxx out of 100 programs it successfully compiled.}
    \label{fig:t-depth-comparison}
\end{figure}

In the early fault-tolerant regime, Clifford gates are relatively inexpensive, whereas non-Clifford rotations must be synthesized into sequences of Clifford and $T$ gates, with each $T$ gate consuming costly distilled magic states~\cite{bravyi2005universal,litinski2019magic}. Consequently, $T$-count governs distillation throughput, while $T$-depth dictates pipeline latency~\cite{amy2013meet,litinski2019game}. We recompile Hamiltonian simulation programs without appending circuit-level peephole optimization to avoid generating numerous single-qubit rotations and then synthesize every $R_Z$ rotation into a Clifford$+T$ sequence using GridSynth~\cite{ross2016optimal} with a precision of $10^{-10}$. 
% We exclude \paulihedral\ and \tetris\ from this evaluation as they are not end-to-end compilers as they only process Pauli strings but do not emit complete circuits for synthesis.

%  ($1.015\times$ for \qiskit, $1.019\times$ for \tket, $0.998\times$ for \quclear, and $1.005\times$ for \phoenix). This invariance is structural. 

As expected, $T$-count remains largely compiler-invariant, staying within 2\% across all compilers, because the non-Clifford rotation count is inherently fixed by the underlying Hamiltonian requiring a median of 102.5 $T$ gates per rotation.
%  and any surplus Clifford rotations introduced by the compilers remain $T$-free. 
Consequently, $T$-depth serves as the primary differentiator between compilers. As shown in \Cref{fig:t-depth-comparison}, \symphony\ reduces $T$-depth by geometric-mean factors of 2.15$\times$ over \qiskit, 2.73$\times$ over \tket, 1.51$\times$ over \paulihedral, 3.04$\times$ over \tetris, 2.30$\times$ over \quclear, and 1.22$\times$ over \phoenix. It produces strictly shallower circuits on 86/100, 85/100, 85/100, 75/100, 78/92, and 62/100 common programs, respectively. Given that the $T$-count is comparable for all compilers, this reduction in depth represents a strict performance improvement.
% These gains are primarily concentrated in condensed-matter Hamiltonians, yielding improvements of up to 10.5$\times$. The gains are comparable for quantum chemistry workloads, where globally supported electronic-structure strings and small vibrational instances alike leave minimal schedulable structure.

\subsection{Ablation Study}

% Superseded by the threshold sweep in \Cref{fig:ablation-emission-threshold}, whose rho = 0 and rho = 1 points reproduce this table.
% \begin{table}[tbp]
%     \centering
%     \caption{Ablation study on emit weight threshold, with the default setting $\rho=0.35$ as the reference.}
%     \label{tab:ablation-emission-threshold}
%     \footnotesize\input{tables/ablation_emission_threshold.tex}
% \end{table}

\begin{figure}[tbp]
    \centering
    \includegraphics[width=\columnwidth]{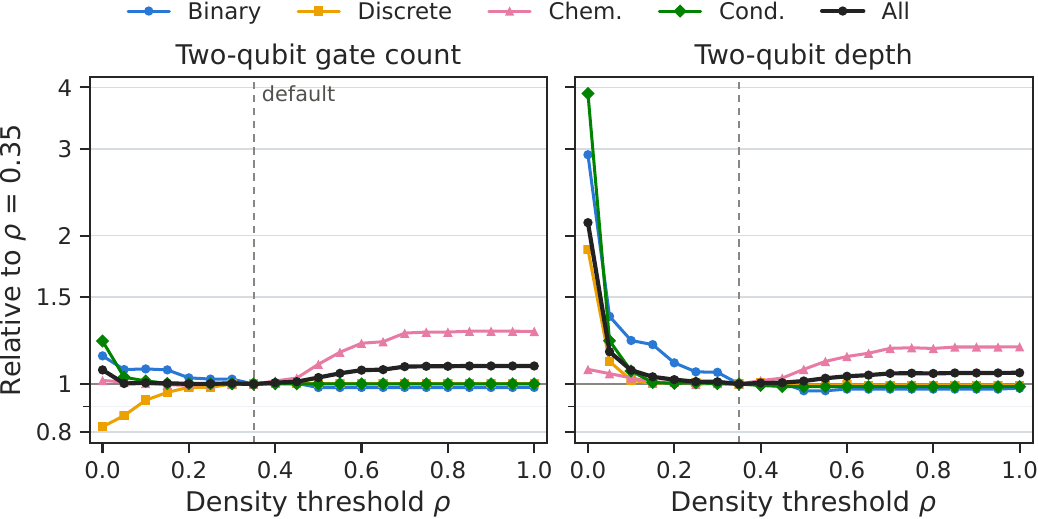}
    \caption{Sensitivity to the emission density threshold $\rho$ on HamLib, swept over $[0,1]$ in steps of $0.05$. Each curve is the geometric mean, over the programs of a category (or all 100), of the per-program two-qubit gate count (left) or depth (right) at threshold $\rho$ divided by that at the default $\rho=0.35$.}
    %  The endpoints $\rho=0$ and $\rho=1$ recover fixed weight-1 and fixed weight-2 emission, respectively.}
    \label{fig:ablation-emission-threshold}
\end{figure}

\subsubsection{Impact of Adaptive Two-qubit Emission}

% The threshold interpolates between the two fixed conventions---$\rho=0$ recovers strict single-qubit emission, $\rho=1$ recovers unconditional weight-2 emission

% Superseded by the paragraph below (endpoint comparison against the former table).
% \Cref{tab:ablation-emission-threshold} validates the adaptive emission mechanism by comparing the default weight density threshold ($\rho=0.35$) against its degenerate endpoints. Strict single-qubit emission ($\rho=0$)---the conventional terminal condition in tableau-based synthesis~\cite{goubault2024faster,yang2025phoenix,kuo2026unified}---proves counterproductive, degrading aggregate two-qubit depth by $112.5\%$ and gate count by $6.8\%$. This penalty is most severe on condensed matter ($+288.5\%$ depth) and binary optimization ($+191.9\%$) instances, where forcing native 2-local and 3-local terms down to weight-1 invariably inflates surrounding operators. Conversely, unconditionally emitting every weight-2 row ($\rho=1$) largely recovers the depth penalty but surrenders $8.8\%$ in aggregate gate count, degrading chemistry workloads by $27.9\%$ where dense strings still reward further simultaneous simplification. Because neither fixed endpoint is universally optimal, dynamically tracking residual tableau density is essential. In practice, \symphony\ can also sweep a range of $\rho$ in parallel to identify the best compilation scheme.

\Cref{fig:ablation-emission-threshold} sweeps the weight density threshold over $\rho\in[0,1]$, normalizing metrics for every program to its result under the default $\rho=0.35$. The degenerate endpoints confirm that neither fixed emission policy suffices. Strict single-qubit emission ($\rho=0$)---the conventional terminal condition in tableau-based synthesis~\cite{goubault2024faster,yang2025phoenix,kuo2026unified}---degrades aggregate two-qubit depth by $112.5\%$ and gate count by $6.8\%$, most severely on condensed matter ($+288.5\%$ depth) and binary optimization ($+191.9\%$), where forcing native 2-local and 3-local terms down to weight-1 inflates surrounding operators; discrete optimization even trades $18.1\%$ fewer gates for $87.5\%$ deeper circuits. Conversely, unconditionally emitting every weight-2 row ($\rho=1$) largely recovers the depth but surrenders $8.8\%$ in aggregate gate count, degrading chemistry by $27.9\%$, where dense strings still reward further simultaneous simplification. The sweep shows that the two failure modes sit at opposite ends of the range: the sparse families lose depth as $\rho$ approaches $0$ (binary optimization already by $5.6\%$ at $\rho=0.3$), whereas chemistry is insensitive below the default but regresses once $\rho$ exceeds $0.4$ ($+9.7\%$ gate count at $\rho=0.5$, approaching $+28\%$ beyond $\rho=0.7$). The default lies in the gap between them: although $\rho=0.35$ is the best grid point on the suite, every threshold in $[0.25,0.45]$ stays within $1.1\%$ of it on both metrics, so the default sits on a broad plateau rather than at a narrowly tuned optimum. Dynamically tracking residual tableau density is therefore what lets a single threshold serve all four families. In practice, \symphony\ can also sweep a range of $\rho$ in parallel to identify the best compilation scheme.

\begin{table}[tbp]
    \centering
    \caption{Ablation study on ASAP and commutativity-optimized scheduling across the HamLib benchmark suite. Trivially ordered circuit results are used as the reference. Only improvements in two-qubit circuit depth are shown.}
    \label{tab:ablation-scheduling}
    \footnotesize% Avg. is the geometric mean of the per-case 2Q-depth ratio
% (ASAP on / ASAP off).  Max. is the largest single-case depth reduction.
% Negative percentages denote a reduction in 2Q circuit depth.
\begin{tabular}{|l|c|c|c|c|}
\hline
\multirow{2}{*}{\diagbox[width=6.5em,height=3.2\baselineskip]{{Cat. (\#)}}{{Improv.}}} & \multicolumn{2}{c|}{ASAP w/o Commute} & \multicolumn{2}{c|}{ASAP w/ Commute} \\
\cline{2-5}
% & \makecell{\makebox[3.8em]{Avg.}} & \makecell{\makebox[3.8em]{Max.}} & \makecell{\makebox[3.8em]{Avg.}} & \makecell{\makebox[3.8em]{Max.}} \\
& \makecell{\makebox[4em]{Avg.}} & \makecell{\makebox[4em]{Max.}} & \makecell{\makebox[4em]{Avg.}} & \makecell{\makebox[4em]{Max.}} \\
\hline
Binary (15) & -49.6\% & -67.0\% & -56.0\% & -74.3\% \\
\hline
Discrete (15) & -51.3\% & -77.0\% & -60.0\% & -77.0\% \\
\hline
Chem. (35) & -19.4\% & -86.7\% & -22.3\% & -86.7\% \\
\hline
Cond. (35) & -83.8\% & -99.1\% & -84.8\% & -99.1\% \\
\hline
\emph{All (100)} & \emph{-60.3\%} & \emph{-99.1\%} & \emph{-63.5\%} & \emph{-99.1\%} \\
\hline
\end{tabular}

\end{table}

\subsubsection{Impact of Causality-Preserving Rescheduling}

% To isolate the algorithmic contribution of our exact commutation calculus, we perform an ablation study comparing standard ASAP scheduling against our commutativity-optimized variant. For each configuration, \Cref{tab:ablation-scheduling} reports the relative two-qubit depth reduction compared to a non-parallelized baseline. Employing conservative ordering dependencies (ASAP w/o Commute) inherently reduces two-qubit depth by an average of $60.3\%$ across the 100 HamLib programs. However, when we relax execution constraints by integrating our precise algebraic rules for commuting move--move and move--block pairs (ASAP w/ Commute), the average depth reduction expands to $63.5\%$.

% Crucially, this enhanced parallelism benefits every single benchmark family. The most substantial average gains occur in the discrete optimization (from $51.3\%$ to $60.0\%$) and binary optimization (from $49.6\%$ to $56.0\%$) suites, where commutativity rules successfully unlock previously blocked gate schedules. Furthermore, while the average reduction in condensed matter appears largely saturated (improving marginally from $83.8\%$ to $84.8\%$), individual instances within this family achieve up to a staggering $99.1\%$ reduction in two-qubit circuit depth. These results empirically validate that our exact commutation rules expose critical scheduling freedom well beyond the capabilities of generic ASAP layering.

% \ZY{Why ASAP works much better than just natural sequence ordering? (mechanism explanation)}

\Cref{tab:ablation-scheduling} isolates the rescheduling pass, reporting two-qubit depth relative to a trivially ordered emission stream. Simply honoring frame-induced precedences (ASAP w/o Commute) removes $60.3\%$ of the depth on average across HamLib, confirming that holistic peeling emits narrowly supported blocks that are largely mutually independent. Layering our exact algebraic commutation relations on top (ASAP w/ Commute) further relaxes conservative dependencies, widening the average depth reduction to $63.5\%$. This improvement spans every benchmark family, shining brightest on discrete ($51.3\%\to60.0\%$) and binary optimization ($49.6\%\to56.0\%$), which feature mutually commuting low-weight terms that purely structural models needlessly serialize. Even in near-saturated condensed-matter workloads ($83.8\%\to84.8\%$), individual instances achieve up to a $99.1\%$ depth reduction. These results validate that exact commutativity analytics successfully unlock scheduling freedom beyond the reach of generic structural ASAP scheduling.

\subsection{Runtime Analysis}

\begin{figure}[tbp]
    \centering
    \includegraphics[width=\columnwidth]{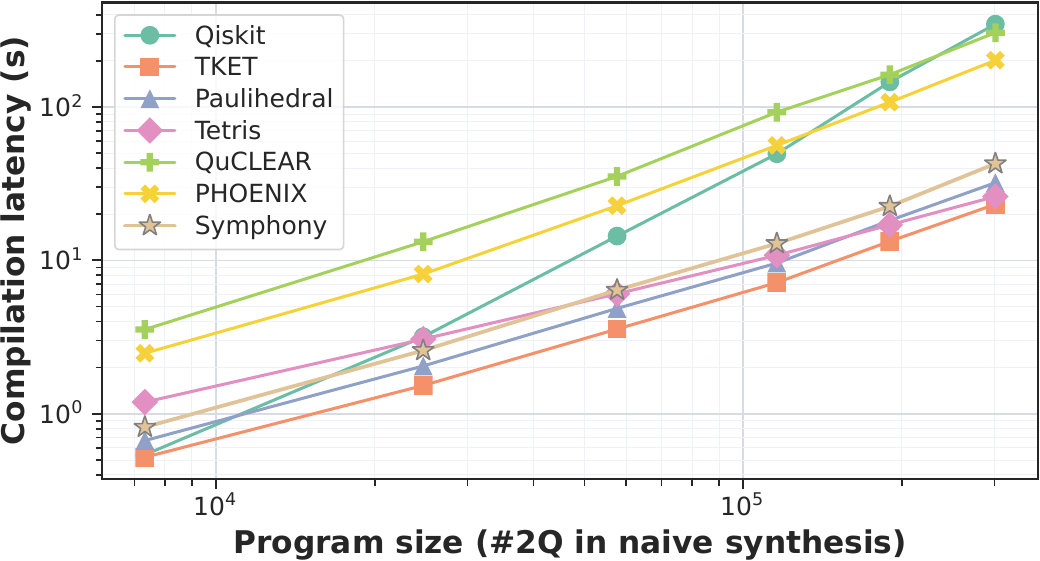}
    \caption{Compilation latency comparison across UCCSD. All compilers generate complete executable circuits, and the latency includes the common circuit-level optimization.}
    \label{fig:compilation-latency}
\end{figure}

To evaluate runtime scalability, we benchmark all compilers across the scaled UCCSD Hamiltonians and report the wall-clock compilation latency in \Cref{fig:compilation-latency}. Overall, \symphony\ exhibits a latency scaling curve that closely aligns with its theoretical polynomial complexity, running substantially faster than \qiskit-\rustiq, \quclear, and \phoenix\ on large-scale Hamiltonians and trailing only \tket-\pcoast, \paulihedral, and \tetris. \symphony's runtime overhead mostly stems from its algebraic simplification and scheduling passes: First, each accepted UCG Clifford triggers a global tableau row update, unlike block-isolated synthesis in \paulihedral\ and \tetris; this becomes especially pronounced for dense, high-weight UCCSD programs. Second, because commutation is non-transitive, each gate must be checked against all preceding operations on overlapping qubits, incurring quadratic checks per qubit despite $O(1)$ pairwise checks. It is worth noting that the core transformation passes of \qiskit-\rustiq\ and \tket-\pcoast\ are accelerated by Rust and C++, respectively. Surprisingly, although the industry-standard, Rust-accelerated \qiskit-\rustiq\ is among the fastest compilers for small-scale programs, it exhibits the poorest scalability when handling large-scale workloads.

% \ZY{complexity comparison for baselines}

\section{Related Work}\label{sec:related_work}

\paragraph{Gate cancellation between Pauli-IR synthesis variants}
A substantial body of work preserves Hamiltonian- or Pauli-level semantics beyond front-end lowering to expose optimization opportunities that are otherwise obscured at the gate level. \citet{tomesh2021optimized} partition Hamiltonian terms into mutually commuting sets and apply graph-coloring- and traveling-salesman-based reordering to balance simulation accuracy with gate cancellation. For 2-local Hamiltonians, 2QAN exploits operator-permutation freedom to co-optimize scheduling, routing, and native-gate selection~\cite{lao20222qan}. Paulihedral introduces a block-wise Pauli intermediate representation and coordinates instruction scheduling with synthesis, gate cancellation, and qubit mapping~\cite{li2022paulihedral}. Tetris further refines the Pauli-string representation to expose two-qubit-gate cancellation and introduces bridge-aware hardware mapping~\cite{jin2024tetris}, while PauliForest combines connectivity-aware synthesis with Pauli-oriented qubit placement~\cite{li2025pauliforest}. Accordion extends this structural approach to end-to-end fermionic compilation by co-designing the fermion-to-qubit mapping, Pauli-string synthesis, and 1D-chain hardware mapping~\cite{gao2026linear}. These approaches primarily optimize term ordering, parity-tree construction, inter-block cancellation, and hardware mapping while retaining the high-level structure of Hamiltonian simulation programs.
% These approaches establish the value of retaining application-level structure, but their principal optimization levers remain term ordering, parity-tree construction, inter-block cancellation, and hardware mapping, rather than algebraically simplifying the complete Pauli sequence under a shared Clifford frame.

% \paulihedral~\cite{li2022paulihedral} and
% \tetris~\cite{jin2024tetris} belong to the same broad Pauli-IR synthesis
% category as \symphony: they preserve Pauli rotations before gate lowering to
% expose cross-rotation cancellation and sharing. \paulihedral\ compiles
% selected blocks, while \tetris\ further co-optimizes synthesis, ordering, and
% routing. In contrast, \symphony\ keeps all rotations in one BSF tableau and
% forms a two-qubit block only after global peeling reduces a row to local
% support.

\paragraph{Graph/Diagram-based synthesis}
Another line of work organizes Pauli and Clifford optimization around commuting clusters, graphical intermediate representations, or extractable circuit regions. Simultaneous diagonalization methods partition Pauli operators into mutually commuting clusters and synthesize a shared Clifford basis transformation for each cluster, avoiding the independent implementation of every Pauli exponential~\cite{van2020circuit,mukhopadhyay2023synthesizing}. Phase-gadget synthesis uses ZX-inspired graphical representations to jointly resynthesize collections of multi-qubit rotations into shallow CNOT networks~\cite{cowtan2019phase}. PCOAST generalizes graph-based Pauli optimization to mixed unitary and non-unitary circuits by representing rotations, preparations, and measurements in a common dependency graph and subsequently resynthesizing the optimized graph with a configurable greedy procedure~\cite{paykin2023pcoast}. QuCLEAR instead extracts Clifford subcircuits toward the circuit boundary and absorbs their effects classically when the surrounding computation permits~\cite{liu2025quclear}. 
% Although some of these frameworks employ Pauli-frame synthesis internally, their primary optimization scope is defined by clusters, dependency graphs, graphical gadgets, or extractable Clifford regions.

% Another type of compiler utilizes ZX-diagram, including \tket~\cite{sivarajah2020t,cowtan2019phase}
% and \pcoast~\cite{paykin2023pcoast}. They
% rewrite Pauli gadgets as ZX diagrams to expose phase interactions before
% circuit extraction. Like \symphony, they optimize relations among Pauli
% rotations rather than synthesize each rotation independently. Their rewrite
% space is diagrammatic; \symphony\ instead uses a BSF tableau and two-qubit
% Clifford moves with target descent.

% \pauliforest~\cite{li2025pauliforest} and 2QAN~\cite{lao20222qan} share the
% same Pauli-network objective but include the hardware interaction graph during
% synthesis. \symphony\ performs BSF simplification in the logical all-to-all
% setting, then uses a downstream mapper to realize the target topology.

\paragraph{Pauli network or tableau-based synthesis}
A closely related family formulates the compilation of Pauli rotations as the construction of a sequence of Clifford transformations on Pauli frames. \citet{schmitz2024graph} represent Pauli frames as graph vertices and Clifford operations as transitions, reducing synthesis to a low-cost frame-path problem. Rustiq greedily constructs a Clifford skeleton that reduces target Paulis to single-qubit rotations by first diagonalizing the current target Pauli operator~\cite{goubault2024faster}. The Clifford Executive Representation of \citet{glos2024generic} jointly compiles multiple, not necessarily commuting Pauli operators while incorporating limited connectivity. PHOENIX instead applies CNOT-equivalent two-qubit Clifford transformations to support-grouped BSF tableaux before ordering the simplified groups~\cite{yang2025phoenix}. pMST defers Clifford operations and uses dependency-constrained, minimum-spanning-tree-guided synthesis over a unified Pauli-rotation tableau targeting generic Clifford+$T$ circuit optimization~\cite{kuo2026unified}. More recent work explores learning~\cite{dubal2025paulinetwork} or longer-horizon search~\cite{machiya2026monteq} policies for the same synthesis problem.
% Rustiq formalizes the related Pauli-network problem: it constructs a Clifford skeleton such that every target Pauli rotation is transformed into a single-qubit rotation at some point in the circuit, and proposes greedy heuristics targeting either entangling-gate count or entangling depth under flexible ordering constraints~\cite{goubault2024faster}. PHOENIX expresses collections of Pauli strings in binary symplectic form and iteratively applies two-qubit Clifford transformations to reduce their column weights before ordering the simplified groups~\cite{yang2025phoenix}. 
% More recent work explores learning~\cite{dubal2025paulinetwork} or longer-horizon search~\cite{machiya2026monteq} policies for the same synthesis problem. 
These methods share a common algebraic foundation, binary symplectic representations updated under Clifford conjugation, but differ in whether optimization is formulated as frame-path construction, Pauli-network growth, tableau simplification, or schedule search.

\smallskip
Finally, additional work addresses the Hamiltonian simulation problem upstream of or orthogonal to the synthesis families above. One line of work alters the Pauli-IR presented to the synthesis engine: product-formula design, randomization, and term ordering trade circuit length against simulation error~\cite{campbell2019random,childs2021theory,tranter2019ordering}; for example, \citet{decker2026kernpiler} partially Trotterize noncommuting terms and directly resynthesizes the resulting small dense unitaries. A second line exploits application-specific structure: UCC compilers partition commuting Pauli exponentials for shared diagonalization and phase-polynomial synthesis~\cite{cowtan2020generic}; QAOA compilers use the commuting interaction layer to co-optimize ordering, placement, and scheduling~\cite{alam2020circuit,lao20222qan}; tile-based compilation exploits spatial symmetry to optimize routing for translationally invariant condensed-matter simulations~\cite{kattemolle2026efficient}; more generally, GadIR preserves the spatial structure of many-body Hamiltonians in its IR for gadget-level reduction and hardware-independent scheduling~\cite{ren2026gadir}; and chemistry-oriented systems optimize fermion-to-qubit encodings or co-design Pauli structure with the target hardware~\cite{liu2024fermihedral,gao2026linear}. A further strand targets $T$-gate overhead directly for early fault-tolerant Hamiltonian simulation~\cite{mukhopadhyay2023synthesizing,li2025non}. Together, these approaches broaden the optimization target beyond Pauli-IR synthesis to encompass Hamiltonian representation, spatial invariants, and fault-tolerant resource cost.

% \symphony belongs to this symplectic Pauli-synthesis family and focuses on global BSF simplification.

% \rustiq~\cite{goubault2024faster}, diagonalization-based synthesis such as
% pMST~\cite{kuo2026unified}, and \symphony\ all use Clifford transformations
% to simplify Pauli networks before lowering the remaining rotations.
% \quclear~\cite{liu2025quclear} similarly extracts and absorbs Clifford
% structure before synthesizing it. \symphony\ differs by scoring native-basis,
% CNOT-equivalent two-qubit Clifford moves on the whole active BSF tableau and
% accepting only moves that reduce a target row. It emits rows at weight two and
% uses commutation information to schedule the emitted blocks.

\section{Conclusion}\label{sec:conclusion}

\symphony\ proves to be an efficient compilation framework that recasts Hamiltonian simulation synthesis as holistic BSF simplification. By applying controlled-Pauli Cliffords as simplification primitives directly across a global BSF tableau, \symphony\ transcends the limitations of restrictive grouping- and diagonalization-based tableau synthesis paradigms. Combined with an adaptive two-qubit block emission mechanism and causality-preserving ASAP scheduling, the framework exploits extensive opportunities for simultaneous simplification and block-level parallelism, unlocking substantial reductions in overall two-qubit gate count and circuit depth.

\begin{acks}
    This research was partially conducted by the AI Chip Center for Emerging Smart Systems (ACCESS), supported by the InnoHK initiative of the Innovation and Technology Commission of the Hong Kong Special Administrative Region Government. It was also supported partially by the Research Grants Council of Hong Kong SAR under Grant No. 16217326, the National Natural Science Foundation of China under Grant No. 92465202, and the Shanghai Institute of Mathematics and Interdisciplinary Sciences under Grant No. SIMIS-ID-2025-QT. Z.~Y.\ would like to thank Xueci Zhang for her helpful suggestions on the paper's visual presentation. D.~D.\ would like to thank God for all of His provisions.
\end{acks}

%%%%%%% -- PAPER CONTENT ENDS -- %%%%%%%%

% use the ACM bibliography style
\bibliographystyle{ACM-Reference-Format}
\bibliography{reference}

%%%%%----- Appendix -----%%%%%%
\appendix

\section{Clifford formalism based on Universal Controlled Gate}\label{sec:appendix_ucg_clifford}

\subsection{Generalized CNOT Gate Definition}

We define the generalized controlled gate associated with a control-axis Pauli operator
$P_1$ on qubit $1$ and a target-axis Pauli operator $P_2$ on qubit $2$ as
\begin{align}
    C_{P_1, P_2}
    =
    \frac{1}{2}(I + P_1) \otimes I
    +
    \frac{1}{2}(I - P_1) \otimes P_2,
\end{align}
where $P_1, P_2 \in \{X, Y, Z\}$.

\subsection{Closed-Form Transformation Rule}

Let $\mathcal{T}_{P_1, P_2}$ denote the Clifford action induced by conjugation with
$C_{P_1, P_2}$. For an arbitrary two-qubit Pauli operator $Q_1 \otimes Q_2$, the transformed
operator is
\begin{align}
    \mathcal{T}_{P_1, P_2}(Q_1 \otimes Q_2)
    =
    \left(Q_1 \cdot P_1^{k_2}\right)
    \otimes
    \left(P_2^{k_1} \cdot Q_2\right),
\end{align}
where the anti-commutation indicators $k_1, k_2 \in \{0,1\}$ are defined by
\begin{align}
    k_1 =
    \begin{cases}
        1, & \text{if } \{Q_1, P_1\} = 0, \\
        0, & \text{otherwise},
    \end{cases}
    \quad
    k_2 =
    \begin{cases}
        1, & \text{if } \{Q_2, P_2\} = 0, \\
        0, & \text{otherwise}.
    \end{cases}
\end{align}
Here, $k_1$ captures the relation between the input operator on the control qubit and the
control reference axis $P_1$, while $k_2$ captures the relation between the input operator on
the target qubit and the target reference axis $P_2$. We use the convention $P^0 = I$, and all
products follow the standard Pauli algebra, including the phase factors $\pm 1$ and $\pm i$.

\subsection{Why the Rule Works}

The generalized CNOT gate is a Clifford operator, so its action on arbitrary Pauli strings is
fully determined by its action on single-qubit Pauli generators. The transformation can be
understood through two complementary propagation rules.
\begin{itemize}[leftmargin=*]
    \item \emph{Control-to-target propagation.} For an operator acting only on the control qubit, we have
    \begin{align}
        C_{P_1, P_2}(Q_1 \otimes I)C_{P_1, P_2}^{\dagger}
        =
        Q_1 \otimes P_2^{k_1}.
    \end{align}
    Therefore, if the input control operator $Q_1$ anti-commutes with $P_1$, the
    conjugation induces $P_2$ on qubit $2$.
    \item \emph{Target-to-control feedback.} For an operator acting only on the target qubit, we have
    \begin{align}
        C_{P_1, P_2}(I \otimes Q_2)C_{P_1, P_2}^{\dagger}
        =
        P_1^{k_2} \otimes Q_2.
    \end{align}
    Therefore, if the input target operator $Q_2$ anti-commutes with $P_2$, the
    conjugation feeds back $P_1$ onto qubit $1$.
\end{itemize}

Combining the two identities gives the closed-form transformation rule above. This expression
uniformly describes both homogeneous and heterogeneous generalized CNOT gates. In practice, the transformation rules can be uniformly applied through symplectic matrix computation or precomputed and stored in a lookup table for efficient application during compilation. The complete set of transformation rules for all nonidentity two-qubit Pauli operators is summarized in \Cref{tab:ucg-pauli-conjugation}.

\subsection{Verification with the Standard CNOT Gate}

Consider the standard CNOT gate, which corresponds to $C_{Z,X}$. In this case,
\begin{align}
    \mathcal{T}_{Z, X}(Q_1 \otimes Q_2)
    =
    \left(Q_1 \cdot Z^{k_2}\right)
    \otimes
    \left(X^{k_1} \cdot Q_2\right).
\end{align}
Representative examples are listed below.
\begin{itemize}[leftmargin=*]
    \item Input $I \otimes X$: $k_1 = 0$ because $I$ commutes with $Z$, and $k_2 = 0$ because
    $X$ commutes with $X$. Hence,
    \[
        \mathcal{T}_{Z,X}(I \otimes X) = I \otimes X.
    \]
    \item Input $X \otimes I$: $k_1 = 1$ because $X$ anti-commutes with $Z$, and $k_2 = 0$
    because $I$ commutes with $X$. Hence,
    \[
        \mathcal{T}_{Z,X}(X \otimes I) = X \otimes X.
    \]
    \item Input $I \otimes Z$: $k_1 = 0$ because $I$ commutes with $Z$, and $k_2 = 1$ because
    $Z$ anti-commutes with $X$. Hence,
    \[
        \mathcal{T}_{Z,X}(I \otimes Z) = Z \otimes Z.
    \]
    \item Input $X \otimes Z$: $k_1 = 1$ and $k_2 = 1$. Therefore,
    \[
        \mathcal{T}_{Z,X}(X \otimes Z)
        =
        (X \cdot Z) \otimes (X \cdot Z)
        =
        (-iY) \otimes (-iY)
        =
        -Y \otimes Y.
    \]
\end{itemize}

These results are fully consistent with the standard CNOT conjugation table and confirm that the
formula correctly captures the transformation behavior in a unified manner.

\begin{table}[tbp]
    \centering
    \caption{Exact UCG Clifford conjugation rules $C_{P_1,P_2}\,M\,C_{P_1,P_2}^\dagger$ for every nonidentity two-qubit Pauli $M=Q_1\otimes Q_2$.}
    \label{tab:ucg-pauli-conjugation}
    \setlength{\tabcolsep}{3.1pt}
    \small\begin{tabular}{|c|ccccccccc|}
    \hline
    $M$ & $C_{XX}$ & $C_{XY}$ & $C_{XZ}$ & $C_{YX}$ & $C_{YY}$ & $C_{YZ}$ & $C_{ZX}$ & $C_{ZY}$ & $C_{ZZ}$ \\
    \hline
    $IX$ & $IX$ & $XX$ & $XX$ & $IX$ & $YX$ & $YX$ & $IX$ & $ZX$ & $ZX$ \\
    $IY$ & $XY$ & $IY$ & $XY$ & $YY$ & $IY$ & $YY$ & $ZY$ & $IY$ & $ZY$ \\
    $IZ$ & $XZ$ & $XZ$ & $IZ$ & $YZ$ & $YZ$ & $IZ$ & $ZZ$ & $ZZ$ & $IZ$ \\
    $XI$ & $XI$ & $XI$ & $XI$ & $XX$ & $XY$ & $XZ$ & $XX$ & $XY$ & $XZ$ \\
    $XX$ & $XX$ & $IX$ & $IX$ & $XI$ & $ZZ$ & $-ZY$ & $XI$ & $-YZ$ & $YY$ \\
    $XY$ & $IY$ & $XY$ & $IY$ & $-ZZ$ & $XI$ & $ZX$ & $YZ$ & $XI$ & $-YX$ \\
    $XZ$ & $IZ$ & $IZ$ & $XZ$ & $ZY$ & $-ZX$ & $XI$ & $-YY$ & $YX$ & $XI$ \\
    $YI$ & $YX$ & $YY$ & $YZ$ & $YI$ & $YI$ & $YI$ & $YX$ & $YY$ & $YZ$ \\
    $YX$ & $YI$ & $-ZZ$ & $ZY$ & $YX$ & $IX$ & $IX$ & $YI$ & $XZ$ & $-XY$ \\
    $YY$ & $ZZ$ & $YI$ & $-ZX$ & $IY$ & $YY$ & $IY$ & $-XZ$ & $YI$ & $XX$ \\
    $YZ$ & $-ZY$ & $ZX$ & $YI$ & $IZ$ & $IZ$ & $YZ$ & $XY$ & $-XX$ & $YI$ \\
    $ZI$ & $ZX$ & $ZY$ & $ZZ$ & $ZX$ & $ZY$ & $ZZ$ & $ZI$ & $ZI$ & $ZI$ \\
    $ZX$ & $ZI$ & $YZ$ & $-YY$ & $ZI$ & $-XZ$ & $XY$ & $ZX$ & $IX$ & $IX$ \\
    $ZY$ & $-YZ$ & $ZI$ & $YX$ & $XZ$ & $ZI$ & $-XX$ & $IY$ & $ZY$ & $IY$ \\
    $ZZ$ & $YY$ & $-YX$ & $ZI$ & $-XY$ & $XX$ & $ZI$ & $IZ$ & $IZ$ & $ZZ$ \\
    \hline
\end{tabular}

\end{table}

\section{Interaction-Rank Criterion for Two-Qubit Pauli-Evolution Blocks}
\label{sec:appendix-interaction-rank}

Here we develop the interaction-rank criterion stated in \Cref{thm:interaction-rank-cnot}. We also give a symbolic alternative to a numerical KAK decomposition for a two-qubit block represented as an exact Pauli evolution.

\subsection{Interaction Matrix}

Let $\vec{\sigma}=(X,Y,Z)^T$.  Every purely nonlocal two-qubit Hamiltonian
can be written as
\begin{align}
    H_J=\sum\nolimits_{\mu,\nu\in\{X,Y,Z\}}J_{\mu\nu}\,
    \sigma_\mu\otimes\sigma_\nu,
    \quad J\in\mathbb{R}^{3\times3},
    \label{eq:interaction-matrix}
\end{align}
where $J$ is its \emph{interaction matrix}.  Its rows and columns are both
ordered as $(X,Y,Z)$; explicitly,
\begin{align}
    J=
    \begin{pmatrix}
        J_{XX} & J_{XY} & J_{XZ}\\
        J_{YX} & J_{YY} & J_{YZ}\\
        J_{ZX} & J_{ZY} & J_{ZZ}
    \end{pmatrix}.
    \label{eq:interaction-matrix-expanded}
\end{align}
Thus, a term $a\,\sigma_\mu\otimes\sigma_\nu$ contributes $a$ to the
entry $J_{\mu\nu}$.

% For example, the block
% \begin{align}
%     H_2=\alpha ZX+\beta YY+\gamma ZZ
%     \label{eq:rank-two-example-hamiltonian}
% \end{align}
% has the interaction matrix
% \begin{align}
%     J_2=
%     \begin{pmatrix}
%         0 & 0 & 0\\
%         0 & \beta & 0\\
%         \gamma & 0 & \alpha
%     \end{pmatrix}.
%     \label{eq:rank-two-example-matrix}
% \end{align}
% The zero $X$ row already makes $\operatorname{rank}(J_2)\leq2$.

\subsection{Local Canonicalization and CNOT Class}

\begin{lemma}[Local action on the interaction matrix]
\label{lem:local-so3-action}
For each $u\in SU(2)$ there is an $R_u\in SO(3)$ satisfying
\begin{align}
    u(\vec r\cdot\vec\sigma)u^\dagger
    =(R_u\vec r)\cdot\vec\sigma.
\end{align}
Consequently, for single-qubit unitaries $u$ and $v$,
\begin{align}
    (u\otimes v)H_J(u^\dagger\otimes v^\dagger)
    =H_{R_u J R_v^T}.
    \label{eq:local-so3-action}
\end{align}
\end{lemma}

\begin{proof}
Conjugation by a single-qubit unitary preserves the Pauli commutation
relations and the Hilbert--Schmidt inner product.  Its action on the three
dimensional real vector space $\operatorname{span}\{X,Y,Z\}$ is therefore a
proper orthogonal transformation.  Applying this action independently to the
two tensor factors gives \eqref{eq:local-so3-action}.
\end{proof}

\begin{lemma}[Proper-SVD Cartan reduction]
\label{lem:proper-svd-cartan}
There are $R_1,R_2\in SO(3)$ and signed singular values $s_1,s_2,s_3$ with
$|s_1|\geq|s_2|\geq|s_3|$ such that
\begin{align}
    R_1JR_2^T=\operatorname{diag}(s_1,s_2,s_3).
    \label{eq:proper-svd}
\end{align}
Therefore, for suitable single-qubit unitaries $u$ and $v$,
\begin{align}
    e^{-iH_J}\sim_{\mathrm{local}}
    e^{-i(s_1XX+s_2YY+s_3ZZ)}.
    \label{eq:canonical-pauli-evolution}
\end{align}
\end{lemma}

\begin{proof}
Start from a real singular-value decomposition of $J$.  If either orthogonal
factor has determinant $-1$, absorb a sign into one diagonal entry, yielding
the proper rotations in \eqref{eq:proper-svd}; this accounts for the possible
sign of one $s_i$.  By \Cref{lem:local-so3-action}, these rotations are
induced by local unitaries.  Finally, $XX$, $YY$, and $ZZ$ commute pairwise,
so exponentiation preserves the local canonical form.
\end{proof}

\begin{proof}[Detailed proof of \Cref{thm:interaction-rank-cnot}]
By \Cref{lem:proper-svd-cartan}, the three Cartan parameters of $U_J$ are
given, up to the usual Weyl-chamber symmetries, by the signed singular values
in \eqref{eq:proper-svd}.  If $\operatorname{rank}(J)\leq2$, then $s_3=0$.
The resulting unitary lies on the $c_3=0$ Cartan boundary, whose CNOT cost is
at most two (apart from local gates).  Conversely, if
$\operatorname{rank}(J)=3$, then $s_3\ne0$.  Away from the measure-zero set
of angles that folds onto a lower Cartan stratum, all three Weyl coordinates
are nonzero and the unitary is in the generic three-CNOT class.  \qedhere
\end{proof}

\subsection{Successive Rotations on the Product Side}

The interaction-rank criterion addresses a block consolidated into a single
exponential $e^{-iH_J}$.  The following lemma and corollary extend the
at-most-two-CNOT guarantee to ordered products of two weight-2 rotations,
whether or not the two terms commute.

\begin{lemma}[Commutators of Pauli bilinears are local]
\label{lem:bilinear-commutator}
For $A=P_a\otimes P_b$ and $B=P_c\otimes P_d$ with $P_\ell\in\{X,Y,Z\}$,
\begin{align}
    [A,B]
    =2i\,\delta_{ac}\,\varepsilon_{bdf}\,(I\otimes P_f)
     +2i\,\varepsilon_{ace}\,\delta_{bd}\,(P_e\otimes I),
    \label{eq:bilinear-commutator}
\end{align}
with repeated Pauli indices summed: two weight-2 Pauli terms either commute
or have a purely local commutator.
\end{lemma}

\begin{proof}
Expand $AB-BA$ with $P_aP_c=\delta_{ac}I+i\varepsilon_{ace}P_e$ on each
qubit; the double-Levi--Civita terms cancel, leaving
\eqref{eq:bilinear-commutator}.
\end{proof}

\begin{corollary}[Successive pair of weight-2 rotations]
\label{cor:successive-pair}
For all angles $\theta_1,\theta_2$ and distinct nontrivial $A=P_a\otimes
P_b$, $B=P_c\otimes P_d$, the product $e^{-i\theta_1A}\,e^{-i\theta_2B}$
lies on the $c_3=0$ boundary and needs at most two CNOTs.
\end{corollary}

\begin{proof}
If $a\ne c$ and $b\ne d$, \Cref{lem:bilinear-commutator} gives $[A,B]=0$, so
the product equals $e^{-i(\theta_1A+\theta_2B)}$, whose interaction matrix
has rank two, and \Cref{thm:interaction-rank-cnot} applies.  If instead
$a=c$ (the case $b=d$ is symmetric), the lemma shows that the Lie algebra
generated by $A$ and $B$ lies in
$P_a\otimes\operatorname{span}\{X,Y,Z\}\oplus
I\otimes\operatorname{span}\{X,Y,Z\}$, whose group consists of
$P_a$-controlled single-qubit unitaries, each locally equivalent to a
controlled-$R_z$ with Weyl coordinates $(c_1,0,0)$.
\end{proof}

\begin{takeaways}[title={Takeaways}]
    \begin{itemize}[leftmargin=*, topsep=2pt, itemsep=2pt, parsep=2pt]
        \item Two-qubit Pauli evolution with interaction rank at most two remains in the $2$-CNOT class.
        \item Two-qubit Pauli evolution with full-rank interaction matrices produce genuine 3-CNOT unitaries.
        \item Noncommutativity between weight-2 Pauli terms is purely local; hence any two successive weight-2 Pauli rotations, commuting or not, stay on the $c_3=0$ boundary and never require three CNOTs (\Cref{cor:successive-pair}).
    \end{itemize}
\end{takeaways}

\subsection{Structural and Graph-Theoretic Criteria}

\begin{lemma}[Axis-span criterion]
\label{lem:axis-span-criterion}
Let
\begin{align}
    L=\operatorname{span}\{\sigma_\mu:\exists\nu,\;J_{\mu\nu}\ne0\},\\
    R=\operatorname{span}\{\sigma_\nu:\exists\mu,\;J_{\mu\nu}\ne0\}.
\end{align}
If $\dim L\leq2$ or $\dim R\leq2$, then
$\operatorname{rank}(J)\leq2$, and hence $e^{-iH_J}$ has a two-CNOT
implementation up to local gates.
\end{lemma}

\begin{proof}
The nonzero rows of $J$ are indexed by the axes in $L$, while the nonzero
columns are indexed by the axes in $R$.  Thus $\operatorname{rank}(J)$ is at
most the number of participating rows and also at most the number of
participating columns.  Either dimension condition bounds the rank by two;
the CNOT conclusion follows from \Cref{thm:interaction-rank-cnot}.
\end{proof}

\begin{corollary}[Determinant test]
\label{cor:determinant-test}
If $\det J\ne0$, then $\operatorname{rank}(J)=3$ and the block is generically
in the three-CNOT class.
\end{corollary}

\begin{proposition}[Support-graph generic-rank criterion]
\label{prop:support-graph-rank}
Suppose that the nonzero entries of $J$ carry algebraically independent
symbolic coefficients.  Form a bipartite graph with left and right vertex
sets both equal to $\{X,Y,Z\}$, and include edge $(\mu,\nu)$ exactly when
$J_{\mu\nu}$ is nonzero.  The generic rank of $J$ equals the maximum matching
number of this graph.
\end{proposition}

\begin{proof}
A $k\times k$ minor has a nonzero determinant polynomial exactly when its
support contains a matching of size $k$: a matching supplies a nonvanishing
permutation monomial, and algebraic independence precludes its cancellation.
The largest such $k$ is both the generic matrix rank and the maximum matching
number.
\end{proof}

Thus a maximum matching of size at most two guarantees the two-CNOT upper
bound for every coefficient choice.  A perfect matching instead guarantees
full rank for generic coefficients and hence the generic three-CNOT class.

\end{document}